\documentclass[reqno]{amsart}

\usepackage{amsmath}
\usepackage{amssymb}
\usepackage{amsthm}
\usepackage{bm}
\usepackage{color}
\usepackage{eucal}
\usepackage{fullpage}
\usepackage{graphicx}
\usepackage{hhline}
\usepackage{lineno}
\usepackage[sc]{mathpazo}
\usepackage{mathrsfs}
\usepackage{mathtools}
\usepackage[numbers,sort&compress,square]{natbib}
\usepackage{setspace}
\usepackage{subfig}
\usepackage{xr}
\usepackage[all,cmtip]{xy}

\newcommand*\patchAmsMathEnvironmentForLineno[1]{%
\expandafter\let\csname old#1\expandafter\endcsname\csname #1\endcsname
\expandafter\let\csname oldend#1\expandafter\endcsname\csname end#1\endcsname
\renewenvironment{#1}%
{\linenomath\csname old#1\endcsname}%
{\csname oldend#1\endcsname\endlinenomath}}%
\newcommand*\patchBothAmsMathEnvironmentsForLineno[1]{%
\patchAmsMathEnvironmentForLineno{#1}%
\patchAmsMathEnvironmentForLineno{#1*}}%
\AtBeginDocument{%
\patchBothAmsMathEnvironmentsForLineno{equation}%
\patchBothAmsMathEnvironmentsForLineno{align}%
\patchBothAmsMathEnvironmentsForLineno{flalign}%
\patchBothAmsMathEnvironmentsForLineno{alignat}%
\patchBothAmsMathEnvironmentsForLineno{gather}%
\patchBothAmsMathEnvironmentsForLineno{multline}%
}

\theoremstyle{definition}

\newtheorem{example}{Example}
\newtheorem{lemma}{Lemma}
\newtheorem{proposition}{Proposition}
\newtheorem{remark}{Remark}
\newtheorem{theorem}{Theorem}

\newcommand{\eq}[1]{\textbf{Eq.~\ref{eq:#1}}}
\newcommand{\fig}[1]{\textbf{Fig.~\ref{fig:#1}}}
\newcommand{\tab}[1]{\textbf{Table~\ref{tab:#1}}}

\title{Reactive learning strategies for iterated games}

\author{Alex McAvoy and Martin A. Nowak}

\begin{document}

\allowdisplaybreaks

\maketitle

\begin{abstract}
In an iterated game between two players, there is much interest in characterizing the set of feasible payoffs for both players when one player uses a fixed strategy and the other player is free to switch. Such characterizations have led to extortionists, equalizers, partners, and rivals. Most of those studies use memory-one strategies, which specify the probabilities to take actions depending on the outcome of the previous round. Here, we consider ``reactive learning strategies," which gradually modify their propensity to take certain actions based on past actions of the opponent. Every linear reactive learning strategy, $\mathbf{p}^{\ast}$, corresponds to a memory one-strategy, $\mathbf{p}$, and vice versa. We prove that for evaluating the region of feasible payoffs against a memory-one strategy, $\mathcal{C}\left(\mathbf{p}\right)$, we need to check its performance against at most $11$ other strategies. Thus, $\mathcal{C}\left(\mathbf{p}\right)$ is the convex hull in $\mathbb{R}^{2}$ of at most $11$ points. Furthermore, if $\mathbf{p}$ is a memory-one strategy, with feasible payoff region $\mathcal{C}\left(\mathbf{p}\right)$, and $\mathbf{p}^{\ast}$ is the corresponding reactive learning strategy, with feasible payoff region $\mathcal{C}\left(\mathbf{p}^{\ast}\right)$, then $\mathcal{C}\left(\mathbf{p}^{\ast}\right)$ is a subset of $\mathcal{C}\left(\mathbf{p}\right)$. Reactive learning strategies are therefore powerful tools in restricting the outcomes of iterated games.
\end{abstract}

\textbf{Keywords:} adaptive strategy; iterated game; memory-one strategy; social dilemma

\section{Introduction}
Since the discovery of zero-determinant strategies for iterated games by \citet{press:PNAS:2012}, there has been a growing interest in the set of possible payoffs that can be achieved against a fixed strategy. Imagine that Alice uses a particular strategy, while Bob can try out any conceivable strategy. The resulting set of payoffs for both Alice and Bob define the ``feasible region'' of Alice's strategy. If Alice uses a so-called zero-determinant strategy \cite{press:PNAS:2012}, then the feasible region is a line. In general, the feasible region is a two-dimensional convex subset of the feasible payoff region of the game (\fig{convexHull}). Using the geometric intuition put forth by \citet{press:PNAS:2012}, subsequent work has explored strategies that generate two-dimensional feasible regions, defined by linear inequalities rather than strict equations \citep{akin:Games:2015,hilbe:GEB:2015,hilbe:NHB:2018}. However, a general description of what this region looks like, as it relates to the type of strategy played, is currently not well-understood. In this study, we characterize the feasible regions for the well-known class of memory-one strategies \cite{nowak:Nature:1993} and consider their relationships to those of a new class of ``reactive learning strategies."

Iterated games have many applications across the social sciences and biology, and with them has come a proliferation of strategy classes of various complexities \citep{axelrod:BB:1984,lehrer:JET:1988,hauert:PRSLB:1997,nowak:Science:2006,hilbe:PNAS:2017}. The type of strategy a player uses for dealing with repeated encounters depends on many factors, including the cognitive capacity of the player and the nature of the underlying ``one-shot" (or ``stage") games. In applications to theoretical biology, the most well-studied type of strategy is known as ``memory-one" because it takes into account the outcome of only the previous encounter when determining how to play in the next round \citep{nowak:Nature:1993,baek:SR:2016}. This class of strategies, while forming only a small subset of all possible ways to play an iterated game \citep{fudenberg:MIT:1991}, has several advantages over more complicated strategies. They permit rich behavior in iterated play, such as punishment for exploitation and reward for cooperation \citep{nowak:Nature:1993,posch:JTB:1999,dalbo:AER:2005,nowak:BP:2006,barlo:JET:2009,dalbo:AER:2011,stewart:SR:2016}; but, owing to their simple memory requirements, they are also straightforward to implement in practice and analyze mathematically.

Memory, however, can apply to more than just the players' actions in the previous round. Since the action a player chooses in any particular encounter is typically chosen stochastically rather than deterministically, a player can also take into account \textit{how} they chose their previous action rather than just the result. In a social dilemma, for instance, each player chooses an action (``cooperate," $C$, or ``defect," $D$) in a given round and receives a payoff for this action against that of the opponent. The distribution with which this action is chosen is referred to as a ``mixed action" and can be specified by a single number between $0$ and $1$, representing the tendency to cooperate. A standard memory-one strategy for player $X$ is given by a five-tuple, $\left(p_{0},p_{CC},p_{CD},p_{DC},p_{DD}\right)$, where $p_{0}$ is the probability of cooperation in the initial round and $p_{xy}$ is the probability of cooperation following an outcome in which $X$ uses action $x$ and the opponent, $Y$, uses action $y$. We consider a variation on this theme, where instead of using $x$ and $y$ to determine the next mixed action, $X$ uses the opponent's action, $y$, to update their own mixed action, $\sigma_{X}\in\left[0,1\right]$, that was used previously to generate $x$. We refer to a strategy of this form as a ``reactive learning strategy."

Such a strategy is ``reactive" because it takes into account the realized action of just the opponent, and it is ``learning" because it adapts to this external stimulus. Like a memory-one strategy, a reactive learning strategy for $X$ requires knowledge of information one round into the past, namely $X$'s mixed action, $\sigma_{X}$, and $Y$'s realized action, $y$. Unlike a memory-one strategy, in which the probability of cooperation is in the set $\left\{p_{0},p_{CC},p_{CD},p_{DC},p_{DD}\right\}$ in every round of the game, a reactive learning strategy can result in a broad range of cooperation tendencies for $X$ over the duration of an iterated game. Moreover, these tendencies can be gradually changed over the course of many rounds, resulting (for example) in high probabilities of cooperation only after the opponent has demonstrated a sufficiently long history of cooperating. Punishment for defection can be similarly realized over a number of interactions. Remembering a probability, $\sigma_{X}$, and an action, $y$, instead of just two actions, $x$ and $y$, can thus lead to more complex behaviors.

This adaptive approach to iterated games is similar to the Bush-Mosteller reinforcement learning algorithm \citep{bush:AMS:1953,roth:GEB:1995,izquierdo:RL:2008}, but there are important distinctions. For one, a reactive learning strategy does not necessarily reinforce behavior resulting in higher payoffs. Furthermore, it completely disregards the focal player's realized action, using only that of the opponent in the update mechanism. But there are certainly reactive learning strategies that are more closely related to reinforcement learning, and we give an example using a variation on the memory-one strategy tit-for-tat (TFT), which we call ``learning tit-for-tat (LTFT)."

In this study, we establish some basic properties of reactive learning strategies relative to the memory-one space. We first characterize the feasible region of a memory-one strategy as the convex hull of at most $11$ points. When then show that there is an embedding of the set of memory-one strategies in the set of reactive learning strategies with the following property: if $\mathbf{p}$ is a memory-one strategy and $\mathbf{p}^{\ast}$ is the corresponding reactive learning strategy, then the feasible region of $\mathbf{p}$ contains the feasible region of $\mathbf{p}^{\ast}$. Moreover, the image of the map $\mathbf{p}\mapsto\mathbf{p}^{\ast}$ is the set of \textit{linear} reactive learning strategies, which consists of those strategies that send a player's mixed action, $\sigma_{X}$, to $\alpha\sigma_{X}+\beta$ for some $\alpha ,\beta\in\left[0,1\right]$. As a consequence, if the goal of a player is to restrict the region of payoffs attainable by the players, then this player should prefer using a linear reactive learning strategy over the corresponding memory-one strategy.

\section{Memory-one strategies}
Consider an iterated game between two players, $X$ and $Y$. In every round, each player chooses an action from the set $\left\{C,D\right\}$ (``cooperate" or ``defect"). They receive payoffs based on the values in the matrix
\begin{align}
\bordermatrix{%
& C & D \cr
C &\ R & \ S \cr
D &\ T & \ P \cr
} .
\end{align}
Over many rounds, these payoffs are averaged to arrive at an expected payoff for each player.

Whereas an action specifies the behavior of a player in one particular encounter, a strategy specifies how a player behaves over the course of many encounters. One of the simplest and best-studied strategies for iterated games is a memory-one strategy \citep{nowak:Nature:1993}, which for player $X$ is defined as follows: for every $\left(x,y\right)\in\left\{C,D\right\}^{2}$ observed as action outcomes of a given round, $X$ devises a mixed action $p_{xy}\in\left[0,1\right]$ for the next round. The notation $p_{xy}$ indicates that this mixed action depends on the (pure) actions of both players in the previous round, not how they arrived at those actions (e.g. by generating an action probabilistically). The term ``strategy" is reserved for the players' behaviors in the iterated game.

Let $\textbf{Mem}_{X}^{1}$ be the space of all memory-one strategies for player $X$ in an iterated game. With just two actions, $C$ and $D$, we have $\textbf{Mem}_{X}^{1}=\left[0,1\right]\times\left[0,1\right]^{4}$, i.e. the space of all $\left(p_{0},p_{CC},p_{CD},p_{DC},p_{DD}\right)\in\left[0,1\right]^{5}$. A pair of memory-one strategies, $\mathbf{p}\coloneqq\left(p_{0},p_{CC},p_{CD},p_{DC},p_{DD}\right)$ and $\mathbf{q}\coloneqq\left(q_{0},q_{CC},q_{CD},q_{DC},q_{DD}\right)$, for $X$ and $Y$, respectively, yield a Markov chain on the space of all action pairs, $\left\{C,D\right\}^{2}$, whose transition matrix is
\begin{align}
M\left(\mathbf{p},\mathbf{q}\right) &= 
\bordermatrix{%
& CC & CD & DC & DD \cr
CC & p_{CC}q_{CC} &\ p_{CC}\left(1-q_{CC}\right) &\ \left(1-p_{CC}\right) q_{CC} &\ \left(1-p_{CC}\right)\left(1-q_{CC}\right) \cr
CD & p_{CD}q_{DC} &\ p_{CD}\left(1-q_{DC}\right) &\ \left(1-p_{CD}\right) q_{DC} &\ \left(1-p_{CD}\right)\left(1-q_{DC}\right) \cr
DC & p_{DC}q_{CD} &\ p_{DC}\left(1-q_{CD}\right) &\ \left(1-p_{DC}\right) q_{CD} &\ \left(1-p_{DC}\right)\left(1-q_{CD}\right) \cr
DD & p_{DD}q_{DD} &\ p_{DD}\left(1-q_{DD}\right) &\ \left(1-p_{DD}\right) q_{DD} &\ \left(1-p_{DD}\right)\left(1-q_{DD}\right) \cr
} \label{eq:memOneMatrix}
\end{align}
and whose initial distribution is $\mu_{0}\coloneqq\left(p_{0}q_{0},p_{0}\left(1-q_{0}\right) ,\left(1-p_{0}\right) q_{0},\left(1-p_{0}\right)\left(1-q_{0}\right)\right)$. If $p_{xy},q_{xy}\in\left(0,1\right)$ for every $x,y\in\left\{C,D\right\}$, then this chain is ergodic and has a unique stationary distribution, $\mu\left(\mathbf{p},\mathbf{q}\right)$, which is independent of $\mu_{0}$. In particular, the expected payoffs, $\pi_{X}\left(\mathbf{p},\mathbf{q}\right) =\mu\left(\mathbf{p},\mathbf{q}\right)\cdot\left(R,S,T,P\right)$ and $\pi_{Y}\left(\mathbf{p},\mathbf{q}\right) =\mu\left(\mathbf{p},\mathbf{q}\right)\cdot\left(R,T,S,P\right)$, are independent of $p_{0}$ and $q_{0}$. In this case, $\pi_{X}$ and $\pi_{Y}$ are functions of just the response probabilities, $\mathbf{p}_{\bullet\bullet}\coloneqq\left(p_{CC},p_{CD},p_{DC},p_{DD}\right)$ and $\mathbf{q}_{\bullet\bullet}\coloneqq\left(q_{CC},q_{CD},q_{DC},q_{DD}\right)$.

A useful way of thinking about a strategy is through its feasible region, i.e. the set of all possible payoff pairs (for $X$ and $Y$) that can be achieved against it. For any memory-one strategy $\mathbf{p}$ of $X$, let
\begin{align}
\mathcal{C}\left(\mathbf{p}\right) &\coloneqq \left\{\left(\pi_{Y}\left(\mathbf{p},\mathbf{q}\right) ,\pi_{X}\left(\mathbf{p},\mathbf{q}\right)\right)\right\}_{\mathbf{q}\in\textbf{Mem}_{X}^{1}}
\end{align}
be this feasible region. (Note that, if $X$ uses a memory-one strategy, then it suffices to assume that $Y$ uses a memory-one strategy by the results of \citet{press:PNAS:2012}.) This subset of the feasible region represents the ``geometry" of strategy $\mathbf{p}$ in the sense that it captures all possible payoff pairs against an opponent.

In this section, we show that the feasible region for $\mathbf{p}\in\textbf{Mem}_{X}^{1}$ with $\mathbf{p}_{\bullet\bullet}\in\left(0,1\right)^{4}$ is characterized by playing $\mathbf{p}$ against the following $11$ strategies: $\left(0,0,0,0\right)$, $\left(0,0,0,1\right)$, $\left(0,0,1,0\right)$, $\left(0,0,1,1\right)$, $\left(0,1,0,1\right)$, $\left(0,1,1,0\right)$, $\left(0,1,1,1\right)$, $\left(1,0,0,1\right)$, $\left(1,0,1,0\right)$, $\left(1,0,1,1\right)$, and $\left(1,1,1,1\right)$. In other words, $\mathcal{C}\left(\mathbf{p}\right)$ is the convex hull of $11$ points (see \fig{convexHull}). Therefore, any $\mathbf{p}\in\textbf{Mem}_{X}^{1}$ generates a simple polygon in $\mathbb{R}^{2}$ whose number of extreme points is uniformly bounded over all game-strategy pairs, $\left(\left(R,S,T,P\right) ,\mathbf{p}\right)$.

\begin{figure}
\centering
\includegraphics[width=0.7\textwidth]{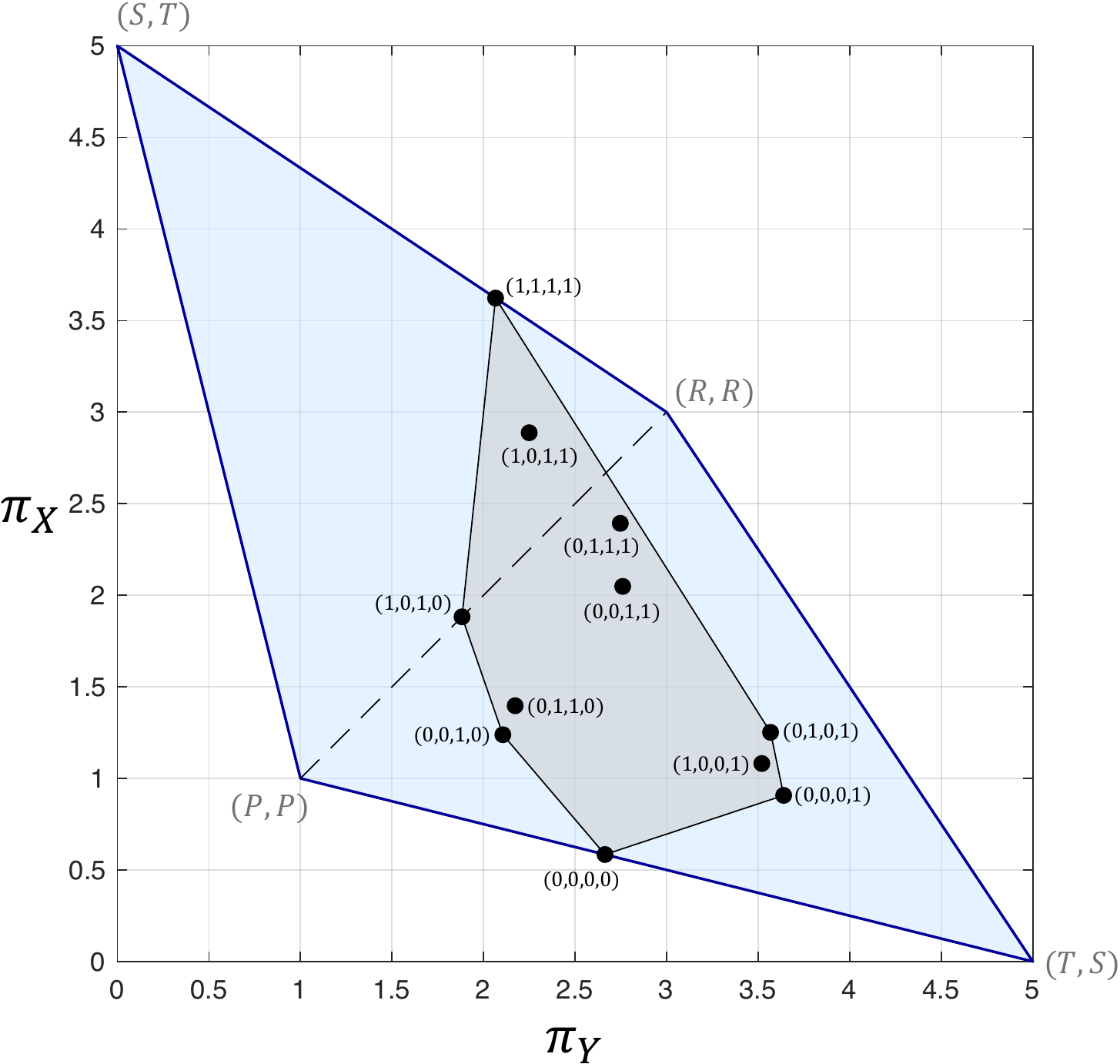}
\caption{\small Feasible region (grey) for a strategy with $\mathbf{p}_{\bullet\bullet}=\left(0.7881,0.8888,0.4686,0.0792\right)$ when $R=3$, $S=0$, $T=5$, and $P=1$. The light blue region depicts the set of all payoff pairs that can be achieved in the iterated game, i.e. the convex hull of the points $\left(R,R\right)$, $\left(S,T\right)$, $\left(P,P\right)$, and $\left(T,S\right)$. The feasible region of $\mathbf{p}$ can be characterized as the convex hull of $11$ points, corresponding to those opponent-strategies, $\mathbf{q}$, appearing next to each black dot. In this instance, five of these points already fall inside of the convex hull of the remaining six. However, one cannot remove one of these $11$ points without destroying this characterization for some game-strategy pair.\label{fig:convexHull}}
\end{figure}

\begin{lemma}\label{lem:oneCoordinateLinear}
For $\mathbf{q}\in\textbf{Mem}_{X}^{1}$ and $x,y\in\left\{C,D\right\}$, let $\left(\mathbf{q};q_{xy}=q_{xy}'\right)$ be the strategy obtained from $\mathbf{q}$ by changing $q_{xy}$ to $q_{xy}'\in\left[0,1\right]$. If $\mathbf{p}_{\bullet\bullet}\in\left(0,1\right)^{4}$, $\mathbf{q}\in\textbf{Mem}_{X}^{1}$, and $x,y\in\left\{C,D\right\}$, then the point $\left(\pi_{Y}\left(\mathbf{p},\mathbf{q}\right) ,\pi_{X}\left(\mathbf{p},\mathbf{q}\right)\right)$ falls on the line joining $\left(\pi_{Y}\left(\mathbf{p},\left(\mathbf{q};q_{xy}=0\right)\right) ,\pi_{X}\left(\mathbf{p},\left(\mathbf{q};q_{xy}=0\right)\right)\right)$ and $\left(\pi_{Y}\left(\mathbf{p},\left(\mathbf{q};q_{xy}=1\right)\right) ,\pi_{X}\left(\mathbf{p},\left(\mathbf{q};q_{xy}=1\right)\right)\right)$.
\end{lemma}
\begin{proof}
Let $\mathbf{p}_{\bullet\bullet}\in\left(0,1\right)^{4}$ and $\mathbf{q}\in\textbf{Mem}_{X}^{1}$. Since the transition matrix of \eq{memOneMatrix} is just $4\times 4$, one can directly solve for its stationary distribution, $\mu\left(\mathbf{p},\mathbf{q}\right)$ (e.g. by using Gaussian elimination or the determinant formula of \citet{press:PNAS:2012}). For example, suppose that $x=y=C$. Then, with
\begin{align}
L\left(q_{CC}\right) &\coloneqq \frac{\left(1-q_{CC}\right)}{1+q_{CC}\left(\frac{\substack{p_{CC} - p_{CC}p_{CD} + p_{CC}p_{DD} - p_{CC}q_{CD} + p_{CC}q_{DD} + p_{DC}q_{CD} - p_{DD}q_{DD} + p_{CC}p_{CD}q_{CD} \\ - p_{CC}p_{CD}q_{DD} - p_{CC}p_{DD}q_{CD} - p_{CD}p_{DC}q_{CD} - p_{CC}p_{DC}q_{DC} + p_{CC}p_{DC}q_{DD} \\ + p_{CC}p_{DD}q_{DC} + p_{CD}p_{DC}q_{DC} - p_{CD}p_{DD}q_{DC} + p_{DC}p_{DD}q_{CD} + p_{CD}p_{DD}q_{DD} - p_{DC}p_{DD}q_{DD} \\ \quad \\ \quad}}{\substack{\quad \\ \quad \\ p_{DD} - p_{CD} - q_{CD} + q_{DD} + p_{CD}q_{CD} + p_{CD}q_{DC} + p_{DC}q_{CD} - p_{CD}q_{DD} - p_{DD}q_{CD} \\ - p_{DC}q_{DC} + p_{DC}q_{DD} + p_{DD}q_{DC} - p_{DD}q_{DD} - p_{CC}p_{CD}q_{DC} - p_{CD}p_{DC}q_{CD} + p_{CD}p_{DC}q_{DC} \\ + p_{CC}p_{DD}q_{DD} + p_{DC}p_{DD}q_{CD} - p_{DC}p_{DD}q_{DD} - p_{CD}q_{CD}q_{DC} + p_{DC}q_{CD}q_{DC} + p_{CD}q_{DC}q_{DD} \\ - p_{DD}q_{DC}q_{DD} + p_{CC}p_{CD}q_{CD}q_{DC} - p_{CC}p_{DC}q_{CD}q_{DC} - p_{CC}p_{CD}q_{DC}q_{DD} + p_{CC}p_{DC}q_{CD}q_{DD} \\ - p_{CC}p_{DD}q_{CD}q_{DD} - p_{CD}p_{DC}q_{CD}q_{DD} - p_{CD}p_{DD}q_{CD}q_{DC} + p_{CD}p_{DD}q_{CD}q_{DD} \\ + p_{CC}p_{DD}q_{DC}q_{DD} + p_{CD}p_{DC}q_{DC}q_{DD} + p_{DC}p_{DD}q_{CD}q_{DC} - p_{DC}p_{DD}q_{DC}q_{DD} + 1}}\right)} ,
\end{align}
one has
\begin{align}
\left(\pi_{Y}\left(\mathbf{p},\mathbf{q}\right) ,\pi_{X}\left(\mathbf{p},\mathbf{q}\right)\right) &= L\left(q_{CC}\right)\left(\pi_{Y}\left(\mathbf{p},\left(\mathbf{q};q_{CC}=0\right)\right) ,\pi_{X}\left(\mathbf{p},\left(\mathbf{q};q_{CC}=0\right)\right)\right) \nonumber \\
&\quad + \left(1-L\left(q_{CC}\right)\right)\left(\pi_{Y}\left(\mathbf{p},\left(\mathbf{q};q_{CC}=1\right)\right) ,\pi_{X}\left(\mathbf{p},\left(\mathbf{q};q_{CC}=1\right)\right)\right) .
\end{align}
Provided $\left(\pi_{Y}\left(\mathbf{p},\left(\mathbf{q};q_{CC}=0\right)\right) ,\pi_{X}\left(\mathbf{p},\left(\mathbf{q};q_{CC}=0\right)\right)\right)\neq\left(\pi_{Y}\left(\mathbf{p},\left(\mathbf{q};q_{CC}=1\right)\right) ,\pi_{X}\left(\mathbf{p},\left(\mathbf{q};q_{CC}=1\right)\right)\right)$, we also have $L\left(0\right) =1$ and $L\left(1\right) =0$. Moreover, one can check that, under this condition, $L'\left(q_{CC}\right)$ is nowhere equal to $0$, and $0\leqslant L\left(q_{CC}\right)\leqslant 1$ for every $q_{CC}\in\left[0,1\right]$. The other cases with $x,y\in\left\{C,D\right\}$ are analogous.
\end{proof}

\begin{remark}
Even when $q_{xy}$ is uniformly distributed between $0$ and $1$, the corresponding points in the feasible region need not be uniformly distributed between the endpoints corresponding to $q_{xy}=0$ and $q_{xy}=1$, respectively (see \fig{nonUniformLine}). This result is therefore somewhat different from the analogous situation of playing against a mixed action in a stage game, where, for a payoff function $u:S_{X}\times S_{Y}\rightarrow\mathbb{R}^{2}$ and mixed action $\sigma_{X}\in\Delta\left(S_{X}\right)$ and $\sigma_{Y}\in\Delta\left(S_{Y}\right)$, one has $u\left(\sigma_{X},\sigma_{Y}\right) =\int_{y\in S_{Y}}u\left(\sigma_{X},y\right)\,d\sigma_{Y}\left(y\right)$ due to linearity.
\end{remark}

\begin{figure}
\centering
\includegraphics[width=0.7\textwidth]{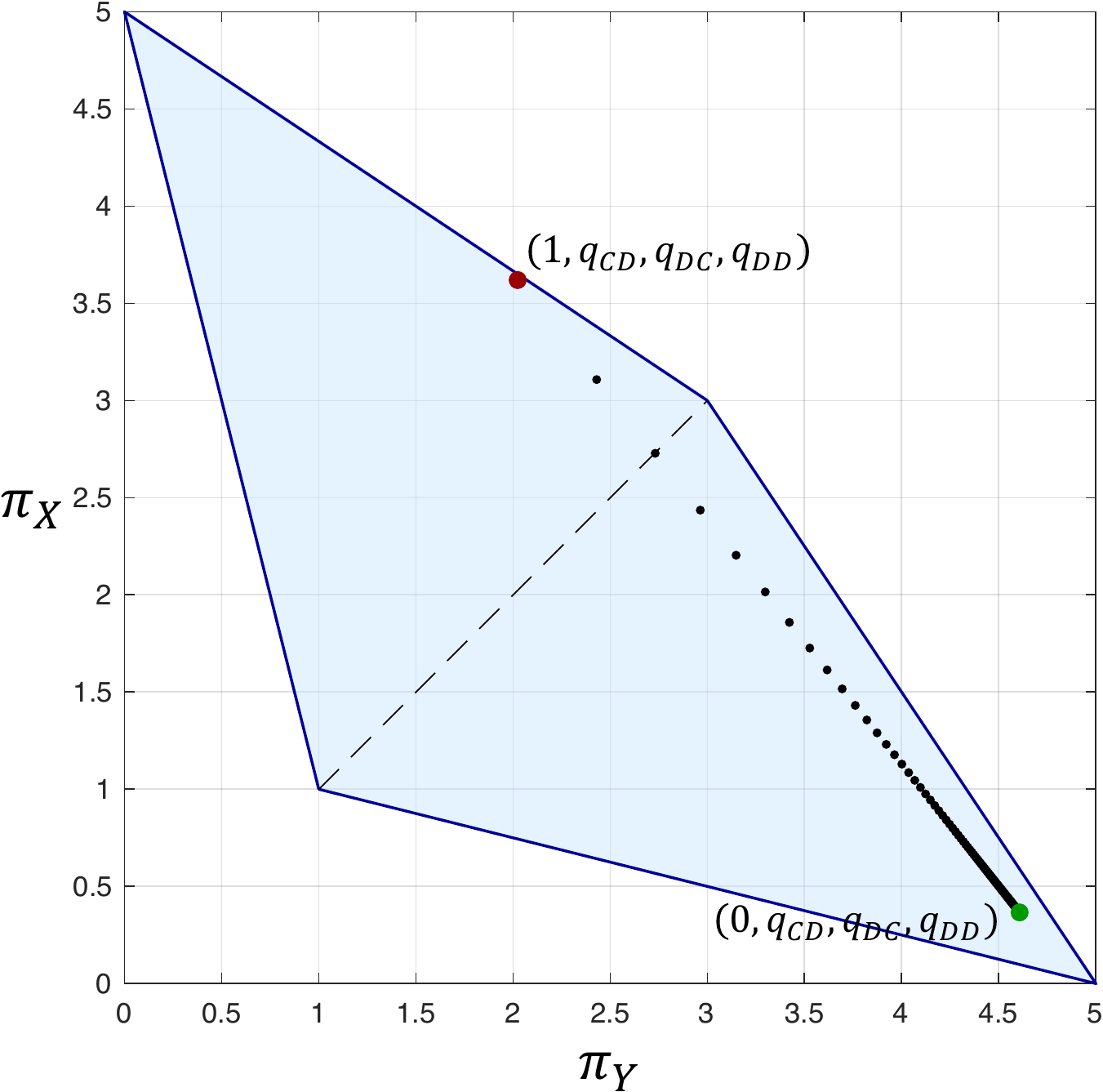}
\caption{\small The set of points $\left(\pi_{Y}\left(\mathbf{p},\mathbf{q}\right) ,\pi_{X}\left(\mathbf{p},\mathbf{q}\right)\right)$, where $\mathbf{p}_{\bullet\bullet}=\left(0.7876,0.9856,0.4095,0.0301\right)$ and $\mathbf{q}_{\bullet\bullet}=\left(q_{CC},0.9963,0.0166,0.9879\right)$ as $q_{CC}$ varies between $0$ (green) and $1$ (red) in uniform increments of $0.01$. The resulting points all fall along a line; however, they are not uniformly distributed even though the distribution of $q_{CC}$ is uniform. Parameters: $R=3$, $S=0$, $T=5$, and $P=1$.\label{fig:nonUniformLine}}
\end{figure}

\begin{proposition}
For any $\mathbf{p}\in\textrm{Mem}_{X}^{1}$ with $\mathbf{p}_{\bullet\bullet}\in\left(0,1\right)^{4}$, $\mathcal{C}\left(\mathbf{p}\right)$ is the convex hull of the following $11$ points:
\begin{subequations}\label{eq:elevenPoints}
\begin{align}
\begin{pmatrix}\pi_{X}^{\left(0,0,0,0\right)} \\ \pi_{Y}^{\left(0,0,0,0\right)}\end{pmatrix} &= \renewcommand*{\arraystretch}{1.3}\begin{pmatrix}\frac{P - Pp_{CD} + Sp_{DD}}{p_{DD} - p_{CD} + 1} \\ \frac{P - Pp_{CD} + Tp_{DD}}{p_{DD} - p_{CD} + 1}\end{pmatrix} ; \\
\begin{pmatrix}\pi_{X}^{\left(0,0,0,1\right)} \\ \pi_{Y}^{\left(0,0,0,1\right)}\end{pmatrix} &= \renewcommand*{\arraystretch}{1.3}\begin{pmatrix}\frac{P + T - Pp_{CD} + Rp_{DD} + Sp_{DC} - Tp_{CD} - Tp_{DD} - Rp_{CD}p_{DD} + Sp_{CC}p_{DD} - Sp_{DC}p_{DD} + Tp_{CD}p_{DD}}{p_{DC} - 2p_{CD} + p_{CC}p_{DD} - p_{DC}p_{DD} + 2} \\ \frac{P + S - Pp_{CD} + Rp_{DD} - Sp_{CD} - Sp_{DD} + Tp_{DC} - Rp_{CD}p_{DD} + Sp_{CD}p_{DD} + Tp_{CC}p_{DD} - Tp_{DC}p_{DD}}{p_{DC} - 2p_{CD} + p_{CC}p_{DD} - p_{DC}p_{DD} + 2}\end{pmatrix} ; \\
\begin{pmatrix}\pi_{X}^{\left(0,0,1,0\right)} \\ \pi_{Y}^{\left(0,0,1,0\right)}\end{pmatrix} &= \renewcommand*{\arraystretch}{1.3}\begin{pmatrix}\frac{P - Pp_{DC} + Sp_{DD} + Tp_{DD} - Pp_{CC}p_{CD} + Pp_{CD}p_{DC} + Rp_{CD}p_{DD} - Tp_{CD}p_{DD}}{2p_{DD} - p_{DC} - p_{CC}p_{CD} + p_{CD}p_{DC} + 1} \\ \frac{P - Pp_{DC} + Sp_{DD} + Tp_{DD} - Pp_{CC}p_{CD} + Pp_{CD}p_{DC} + Rp_{CD}p_{DD} - Sp_{CD}p_{DD}}{2p_{DD} - p_{DC} - p_{CC}p_{CD} + p_{CD}p_{DC} + 1}\end{pmatrix} ; \\
\begin{pmatrix}\pi_{X}^{\left(0,0,1,1\right)} \\ \pi_{Y}^{\left(0,0,1,1\right)}\end{pmatrix} &= \renewcommand*{\arraystretch}{1.8}\begin{pmatrix}\frac{\substack{P + T - Pp_{DC} + Rp_{DD} + Sp_{DC} - Tp_{DD} - Pp_{CC}p_{CD} + Pp_{CD}p_{DC} + Rp_{CD}p_{DC} \\ - Rp_{DC}p_{DD} + Sp_{CC}p_{DD} - Sp_{DC}p_{DD} - Tp_{CC}p_{CD} + Tp_{CC}p_{DD}}}{2\left(p_{CC}p_{DD} - p_{CC}p_{CD} + p_{CD}p_{DC} - p_{DC}p_{DD} + 1\right)} \\
\frac{\substack{P + S - Pp_{DC} + Rp_{DD} - Sp_{DD} + Tp_{DC} - Pp_{CC}p_{CD} + Pp_{CD}p_{DC} + Rp_{CD}p_{DC} \\ - Rp_{DC}p_{DD} - Sp_{CC}p_{CD} + Sp_{CC}p_{DD} + Tp_{CC}p_{DD} - Tp_{DC}p_{DD}}}{2\left(p_{CC}p_{DD} - p_{CC}p_{CD} + p_{CD}p_{DC} - p_{DC}p_{DD} + 1\right)}\end{pmatrix} ; \\
\begin{pmatrix}\pi_{X}^{\left(0,1,0,1\right)} \\ \pi_{Y}^{\left(0,1,0,1\right)}\end{pmatrix} &= \renewcommand*{\arraystretch}{1.3}\begin{pmatrix}\frac{T + Pp_{DC} + Rp_{DC} - Tp_{CD} - Tp_{DD} - Pp_{CD}p_{DC} - Rp_{CD}p_{DC} + Sp_{CC}p_{DC} + Tp_{CD}p_{DD}}{2p_{DC} - p_{CD} - p_{DD} + p_{CC}p_{DC} - 2p_{CD}p_{DC} + p_{CD}p_{DD} + 1} \\ \frac{S + Pp_{DC} + Rp_{DC} - Sp_{CD} - Sp_{DD} - Pp_{CD}p_{DC} - Rp_{CD}p_{DC} + Sp_{CD}p_{DD} + Tp_{CC}p_{DC}}{2p_{DC} - p_{CD} - p_{DD} + p_{CC}p_{DC} - 2p_{CD}p_{DC} + p_{CD}p_{DD} + 1}\end{pmatrix} ; \\
\begin{pmatrix}\pi_{X}^{\left(0,1,1,0\right)} \\ \pi_{Y}^{\left(0,1,1,0\right)}\end{pmatrix} &= \renewcommand*{\arraystretch}{1.3}\begin{pmatrix}\frac{Pp_{DC} + Tp_{DD} - Pp_{CC}p_{DC} + Rp_{DC}p_{DD} + Sp_{DC}p_{DD} - Tp_{CD}p_{DD}}{p_{DC} + p_{DD} - p_{CC}p_{DC} - p_{CD}p_{DD} + 2p_{DC}p_{DD}} \\ \frac{Pp_{DC} + Sp_{DD} - Pp_{CC}p_{DC} + Rp_{DC}p_{DD} - Sp_{CD}p_{DD} + Tp_{DC}p_{DD}}{p_{DC} + p_{DD} - p_{CC}p_{DC} - p_{CD}p_{DD} + 2p_{DC}p_{DD}}\end{pmatrix} ; \\
\begin{pmatrix}\pi_{X}^{\left(0,1,1,1\right)} \\ \pi_{Y}^{\left(0,1,1,1\right)}\end{pmatrix} &= \renewcommand*{\arraystretch}{1.3}\begin{pmatrix}\frac{T + Pp_{DC} + Rp_{DC} - Tp_{DD} - Pp_{CC}p_{DC} + Sp_{CC}p_{DC} - Tp_{CC}p_{CD} + Tp_{CC}p_{DD}}{2p_{DC} - p_{DD} - p_{CC}p_{CD} + p_{CC}p_{DD} + 1} \\ \frac{S + Pp_{DC} + Rp_{DC} - Sp_{DD} - Pp_{CC}p_{DC} - Sp_{CC}p_{CD} + Sp_{CC}p_{DD} + Tp_{CC}p_{DC}}{2p_{DC} - p_{DD} - p_{CC}p_{CD} + p_{CC}p_{DD} + 1}\end{pmatrix} ; \\
\begin{pmatrix}\pi_{X}^{\left(1,0,0,1\right)} \\ \pi_{Y}^{\left(1,0,0,1\right)}\end{pmatrix} &= \renewcommand*{\arraystretch}{1.3}\begin{pmatrix}-\frac{P + T - Pp_{CC} - Pp_{CD} + Rp_{DD} + Sp_{DC} - Tp_{CC} - Tp_{CD} + Pp_{CC}p_{CD} - Rp_{CD}p_{DD} - Sp_{CC}p_{DC} + Tp_{CC}p_{CD}}{2p_{CC} + 2p_{CD} - p_{DC} - p_{DD} - 2p_{CC}p_{CD} + p_{CC}p_{DC} + p_{CD}p_{DD} - 2} \\ -\frac{P + S - Pp_{CC} - Pp_{CD} + Rp_{DD} - Sp_{CC} - Sp_{CD} + Tp_{DC} + Pp_{CC}p_{CD} - Rp_{CD}p_{DD} + Sp_{CC}p_{CD} - Tp_{CC}p_{DC}}{2p_{CC} + 2p_{CD} - p_{DC} - p_{DD} - 2p_{CC}p_{CD} + p_{CC}p_{DC} + p_{CD}p_{DD} - 2}\end{pmatrix} ; \\
\begin{pmatrix}\pi_{X}^{\left(1,0,1,0\right)} \\ \pi_{Y}^{\left(1,0,1,0\right)}\end{pmatrix} &= \renewcommand*{\arraystretch}{1.3}\begin{pmatrix}\frac{P - Pp_{CC} - Pp_{DC} + Sp_{DD} + Tp_{DD} + Pp_{CC}p_{DC} + Rp_{CD}p_{DD} - Sp_{CC}p_{DD} - Tp_{CC}p_{DD}}{2p_{DD} - p_{DC} - p_{CC} + p_{CC}p_{DC} - 2p_{CC}p_{DD} + p_{CD}p_{DD} + 1} \\ \frac{P - Pp_{CC} - Pp_{DC} + Sp_{DD} + Tp_{DD} + Pp_{CC}p_{DC} + Rp_{CD}p_{DD} - Sp_{CC}p_{DD} - Tp_{CC}p_{DD}}{2p_{DD} - p_{DC} - p_{CC} + p_{CC}p_{DC} - 2p_{CC}p_{DD} + p_{CD}p_{DD} + 1}\end{pmatrix} ; \\
\begin{pmatrix}\pi_{X}^{\left(1,0,1,1\right)} \\ \pi_{Y}^{\left(1,0,1,1\right)}\end{pmatrix} &= \renewcommand*{\arraystretch}{1.3}\begin{pmatrix}\frac{P + T - Pp_{CC} - Pp_{DC} + Rp_{DD} + Sp_{DC} - Tp_{CC} + Pp_{CC}p_{DC} + Rp_{CD}p_{DC} - Rp_{DC}p_{DD} - Sp_{CC}p_{DC}}{p_{DD} - 2p_{CC} + p_{CD}p_{DC} - p_{DC}p_{DD} + 2} \\ \frac{P + S - Pp_{CC} - Pp_{DC} + Rp_{DD} - Sp_{CC} + Tp_{DC} + Pp_{CC}p_{DC} + Rp_{CD}p_{DC} - Rp_{DC}p_{DD} - Tp_{CC}p_{DC}}{p_{DD} - 2p_{CC} + p_{CD}p_{DC} - p_{DC}p_{DD} + 2}\end{pmatrix} ; \\
\begin{pmatrix}\pi_{X}^{\left(1,1,1,1\right)} \\ \pi_{Y}^{\left(1,1,1,1\right)}\end{pmatrix} &= \renewcommand*{\arraystretch}{1.3}\begin{pmatrix}\frac{T + Rp_{DC} - Tp_{CC}}{p_{DC} - p_{CC} + 1} \\ \frac{S + Rp_{DC} - Sp_{CC}}{p_{DC} - p_{CC} + 1}\end{pmatrix} .
\end{align}
\end{subequations}
\end{proposition}
\begin{proof}
\citet{press:PNAS:2012} show that if $X$ uses a memory-one strategy, $\mathbf{p}$, then any strategy of the opponent, $\mathbf{y}$, can be replaced by a memory-one strategy, $\mathbf{q}$, without changing the payoffs to $X$ and $Y$; thus, if $X$ uses a memory-one strategy, one may assume without a loss of generality that $Y$ also uses a memory-one strategy. If $\mathbf{p}_{\bullet\bullet}\in\left(0,1\right)^{4}$ and $\mathbf{q}\in\textbf{Mem}_{X}^{1}$, the fact that $\left(\pi_{Y}\left(\mathbf{p},\mathbf{q}\right) ,\pi_{X}\left(\mathbf{p},\mathbf{q}\right)\right)$ can be written as a convex combination of the $16$ points $\left\{\left(\pi_{Y}\left(\mathbf{p},\mathbf{q}'\right) ,\pi_{X}\left(\mathbf{p},\mathbf{q}'\right)\right)\right\}_{\mathbf{q}_{\bullet\bullet}'\in\left\{0,1\right\}^{4}}$ then follows immediately from Lemma~\ref{lem:oneCoordinateLinear}. Moreover, the points corresponding to $\left(0,0,0,0\right)$, $\left(0,1,0,0\right)$, and $\left(1,0,0,0\right)$ are the same, as are the points corresponding to $\left(1,1,0,1\right)$, $\left(1,1,1,0\right)$, and $\left(1,1,1,1\right)$; thus, we can eliminate four points. Furthermore, we can remove the point associated to $\left(1,1,0,0\right)$ because it lies on the line connecting the points associated to $\left(0,0,0,0\right)$ and $\left(1,1,1,1\right)$. One can easily check that the remaining $11$ points have the following property: if point $i$ is removed, then there exist $R,S,T,P$ and $\mathbf{p}$ for which $\mathcal{C}\left(\mathbf{p}\right)$ is not the convex hull of the $10$ points different from $i$ (\tab{11points}). Thus, for a general $\mathbf{p}$ and payoff matrix, all $11$ of these points are required.
\begin{table}
\centering
\bgroup
\def\arraystretch{1.3}
\begin{tabular}{|c|c|c|}
\hline
\textbf{\textit{point}} & $\begin{pmatrix}\bm{R} & \bm{S} \\ \bm{T} & \bm{P}\end{pmatrix}$ & $\mathbf{p}_{\bullet\bullet}$ \\
\hhline{|=|=|=|}
$\bm{\pi_{X,Y}^{\left(0,0,0,0\right)}}$ & $\begin{pmatrix}4.5953 & -3.5001 \\ -0.1798 & 4.4972\end{pmatrix}$ & $\left(0.0347,0.8913,0.9873,0.1164\right)$ \\
\hline
$\bm{\pi_{X,Y}^{\left(0,0,0,1\right)}}$ & $\begin{pmatrix}3.5909 & 3.7183 \\ 3.1091 & 2.6508\end{pmatrix}$ & $\left(0.3420,0.5591,0.0468,0.9941\right)$ \\
\hline
$\bm{\pi_{X,Y}^{\left(0,0,1,0\right)}}$ & $\begin{pmatrix}0.1150 & 1.2677 \\ -2.8725 & 1.4290\end{pmatrix}$ & $\left(0.8937,0.9211,0.6995,0.0052\right)$ \\
\hline
$\bm{\pi_{X,Y}^{\left(0,0,1,1\right)}}$ & $\begin{pmatrix}-0.1523 & 1.7642 \\ -3.3334 & -3.9907\end{pmatrix}$ & $\left(0.5319,0.4107,0.9805,0.0823\right)$ \\
\hline
$\bm{\pi_{X,Y}^{\left(0,1,0,1\right)}}$ & $\begin{pmatrix}2.1084 & 0.4235 \\ 4.5449 & -4.5716\end{pmatrix}$ & $\left(0.3897,0.6428,0.2422,0.0300\right)$ \\
\hline
$\bm{\pi_{X,Y}^{\left(0,1,1,0\right)}}$ & $\begin{pmatrix}2.5627 & -2.5701 \\ -4.1353 & 4.0437\end{pmatrix}$ & $\left(0.7502,0.7603,0.9999,0.3161\right)$ \\
\hline
$\bm{\pi_{X,Y}^{\left(0,1,1,1\right)}}$ & $\begin{pmatrix}0.0600 & 1.1524 \\ 2.8660 & 1.3631\end{pmatrix}$ & $\left(0.1145,0.9494,0.7587,0.9214\right)$ \\
\hline
$\bm{\pi_{X,Y}^{\left(1,0,0,1\right)}}$ & $\begin{pmatrix}-4.4025 & 1.6813 \\ -2.9162 & 1.1664\end{pmatrix}$ & $\left(0.9629,0.0020,0.2554,0.8444\right)$ \\
\hline
$\bm{\pi_{X,Y}^{\left(1,0,1,0\right)}}$ & $\begin{pmatrix}0.1167 & 2.5125 \\ -0.3462 & -4.6919\end{pmatrix}$ & $\left(0.4121,0.4373,0.5380,0.8915\right)$ \\
\hline
$\bm{\pi_{X,Y}^{\left(1,0,1,1\right)}}$ & $\begin{pmatrix}-0.3787 & 1.1357 \\ 1.5417 & 2.7617\end{pmatrix}$ & $\left(0.2570,0.5191,0.1293,0.9332\right)$ \\
\hline
$\bm{\pi_{X,Y}^{\left(1,1,1,1\right)}}$ & $\begin{pmatrix}-1.8211 & -3.2300 \\ -4.6281 & -0.4609\end{pmatrix}$ & $\left(0.0009,0.4996,0.4362,0.9653\right)$ \\
\hline
\end{tabular}
\egroup
\caption{For each point, $\bm{\pi_{X,Y}^{\left(i_{1},i_{2},i_{3},i_{3}\right)}}$, the feasible region $\mathcal{C}\left(\mathbf{p}\right)$ cannot (in general) be expressed as the convex hull of the remaining $10$ points different from $\bm{\pi_{X,Y}^{\left(i_{1},i_{2},i_{3},i_{3}\right)}}$. That is, each row gives \textit{(i)} one of the $11$ points of which $\mathcal{C}$ is the convex hull and \textit{(ii)} an example of a game-strategy pair for which $\bm{\pi_{X,Y}^{\left(i_{1},i_{2},i_{3},i_{3}\right)}}$ is an extreme point of $\mathcal{C}\left(\mathbf{p}\right)$.\label{tab:11points}}
\end{table}
\end{proof}

\begin{remark}
$\mathbf{p}$ enforces a linear payoff relationship if and only if these $11$ points are collinear.
\end{remark}

\begin{remark}
One needs all $11$ of these points for general $R,S,T,P$ and $\mathbf{p}$. However, for any particular game-strategy pair, it is often the case that several of these points are unnecessary because they lie within the convex hull of some other subset of these $11$ points; they are typically not all extreme points of $\mathcal{C}\left(\mathbf{p}\right)$.
\end{remark}

\section{Reactive learning strategies}
In a traditional memory-one strategy, $X$'s probability of playing $C$ depends on the realized actions of the two players, $x$ and $y$. However, $X$ can observe more than just their pure action against the opponent's; they also know how they arrived at $x$ (i.e. they know the mixed action, $\sigma_{X}$, that resulted in $x$ in the previous round). Of course, $X$ need not be able to see $Y$'s mixed action, but they can still observe the pure action $Y$ played. Therefore, an alternative notion of a memory-one strategy for player $X$ could be defined as follows: after $X$ plays $\sigma_{X}\in\left[0,1\right]$ and $Y$ plays $y$, $X$ then chooses a new action based on the distribution $p_{\sigma_{X}y}^{\ast}\in\left[0,1\right]$. In this formulation, $p^{\ast}$ is a map from $\left[0,1\right]\times\left\{C,D\right\}$ to $\left[0,1\right]$. We refer to such a map, $p^{\ast}$, together with $X$'s initial probability of playing $C$, $p_{0}$, as a ``reactive learning strategy" for player $X$ (\fig{strategySpaces}).

\begin{figure}
\centering
\includegraphics[width=0.7\textwidth]{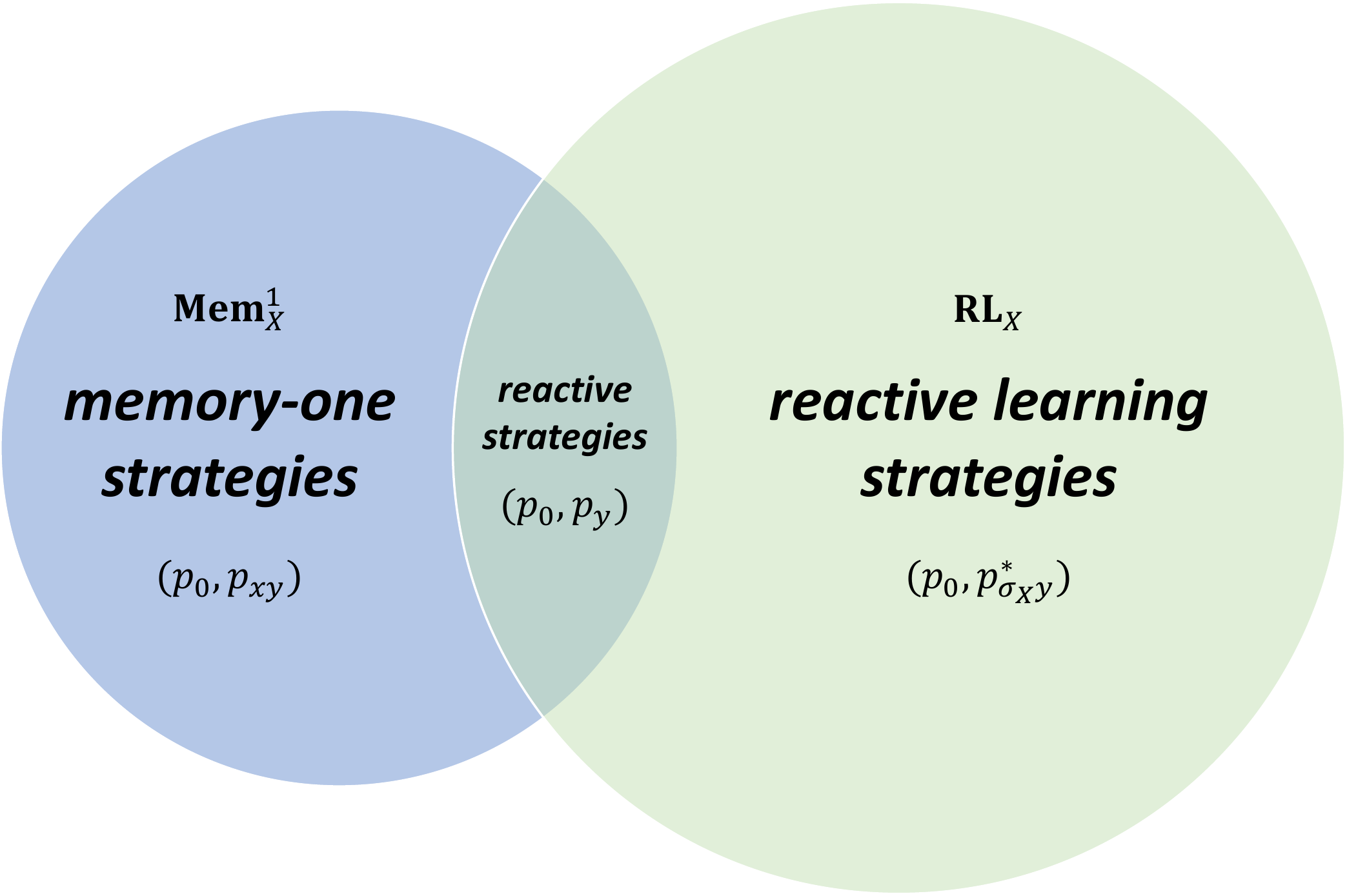}
\caption{\small The space of memory-one strategies, $\textbf{Mem}_{X}^{1}$, as it relates to the space of reactive learning strategies, $\mathbf{RL}_{X}$. Both sets contain the space of reactive strategies \citep{nowak:AAM:1990}, which take into account only the last move, $y$, of the opponent. Whereas a memory-one strategy takes into account the last pure action of $X$ as well, $x$, a reactive learning strategy uses $X$'s last \textit{mixed} action, $\sigma_{X}\in\left[0,1\right]$. After each round, a reactive learning strategy uses $y$ to update $X$'s probability of cooperating. $\textbf{RL}_{X}$ is ``larger" than $\textbf{Mem}_{X}^{1}$ in the sense that there is an injective map $\textbf{Mem}_{X}^{1}\rightarrow\textbf{RL}_{X}$ that is not surjective.\label{fig:strategySpaces}}
\end{figure}

In other words, in contrast to $\textbf{Mem}_{X}^{1}=\left[0,1\right]\times\left[0,1\right]^{4}$, which can be alternatively described as
\begin{align}
\textbf{Mem}_{X}^{1} &= \left[0,1\right]\times\Big\{ p : \left\{C,D\right\}\times\left\{C,D\right\} \rightarrow \left[0,1\right] \Big\} ,
\end{align}
we define the space of reactive learning strategies as
\begin{align}
\textbf{RL}_{X} &\coloneqq \left[0,1\right]\times\Big\{ p^{\ast} : \left[0,1\right]\times\left\{C,D\right\} \rightarrow \left[0,1\right] \Big\} ,
\end{align}
where $\left[0,1\right]$ indicates the space of mixed actions for $X$ and $\left\{C,D\right\}$ indicates the action space for $Y$. Although $\left[0,1\right]$ is a much larger space than $\left\{C,D\right\}$, the updates of mixed actions can be easier to specify using reactive learning strategies since they allow for adaptive modification of an existing mixed action (without the need to devise a new mixed action from scratch after every observed history of play).

\begin{example}\label{ex:eta}
Suppose that player $X$ starts by playing $C$ and $D$ with equal probability, i.e. $p_{0}=1/2$. For fixed $\eta\in\left[0,1\right]$ (the ``learning rate"), cooperation from the opponent leads to $p_{\sigma_{X} C}^{\ast}=\left(1-\eta\right)\sigma_{X}+\eta$ while defection leads to $p_{\sigma_{X} D}^{\ast}=\left(1-\eta\right)\sigma_{X}$. Thus, a long pattern of exploitation by $Y$ leads $X$ to defect more often. On the other hand, $X$ does not immediately forgive such behavior but rather requires $Y$ to cooperate repeatedly to bring $X$ back up to higher levels of cooperation. For example, if $X$ starts with $p_{0}$ and $Y$ defects $\ell$ times in a row, then $X$ subsequently cooperates with probability $\left(1-\eta\right)^{\ell}p_{0}$. In order to bring $X$'s probability of cooperation above $p_{0}$ once again, $Y$ must then cooperate for $T$ rounds, where
\begin{align}
T &\geqslant \frac{\log\left(\frac{1-p_{0}}{1-\left(1-\eta\right)^{\ell}p_{0}}\right)}{\log\left(1-\eta\right)} .
\end{align}
We refer to this strategy as ``learning tit-for-tat (LTFT)" because it pushes a player's cooperation probability in the direction of the opponent's last move (see \fig{LTFT}). In this way, a reactive learning strategy can encode more complicated behavior than a memory-one strategy. Conversely, memory-one strategies can also encode behavior not captured by reactive learning strategies, which we discuss further in \S\ref{sec:feasibleRL}.
\end{example}

\begin{figure}
\centering
\includegraphics[width=0.5\textwidth]{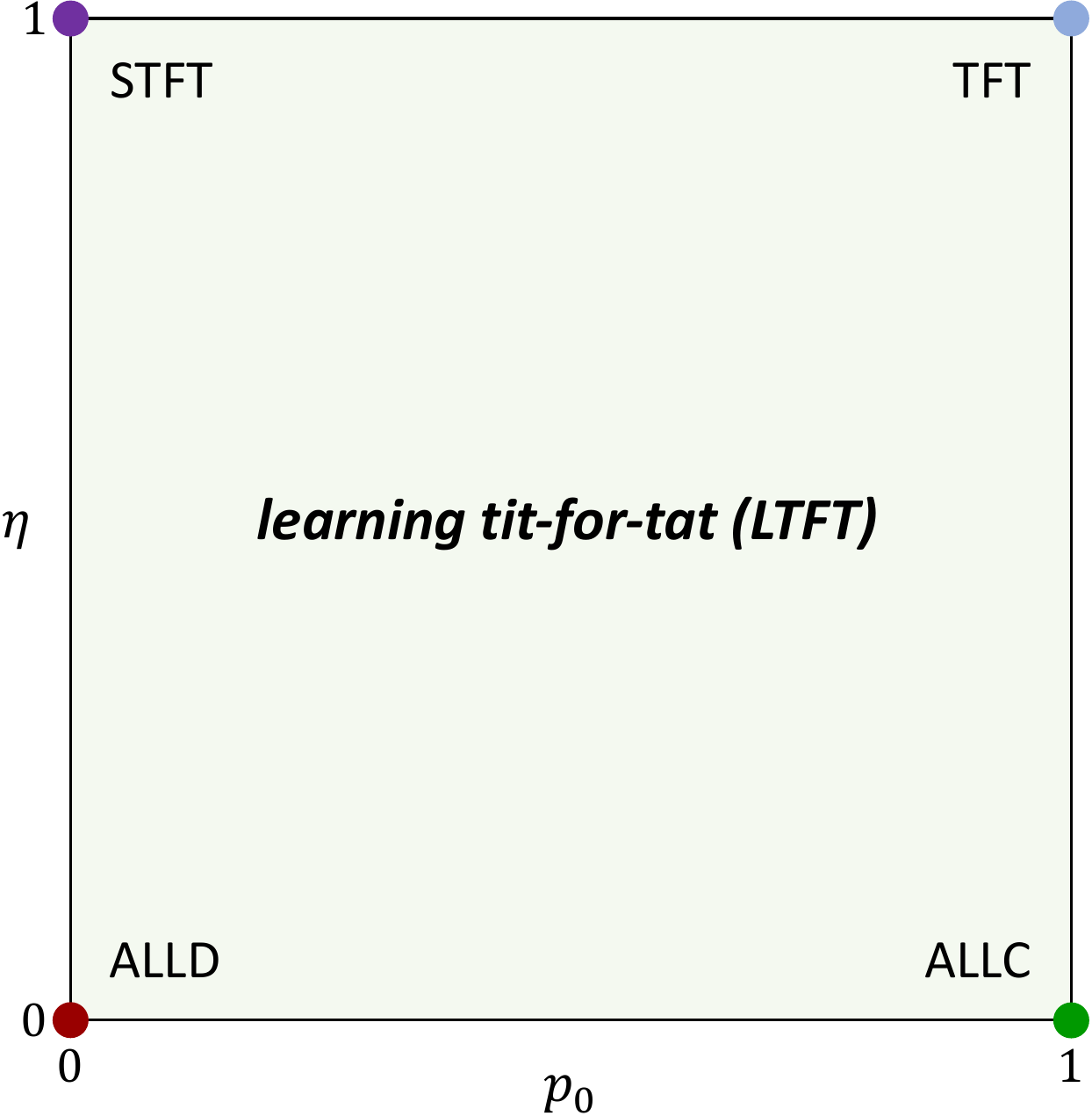}
\caption{\small ``Learning tit-for-tat (LTFT)," an analogue of tit-for-tat (TFT) within the space of reactive learning strategies. LTFT is the function of two parameters, $p_{0}$ (the initial mixed action) and $\eta$ (the learning rate). Player $X$ initially plays $C$ with probability $p_{0}$. In all subsequent rounds, if $X$ played $C$ with probability $\sigma_{X}$ and $Y$ played $C$ (resp. $D$) in the previous round, in the next round $X$ plays $C$ with probability $p_{\sigma_{X}C}^{\ast}=\left(1-\eta\right)\sigma_{X}+\eta$ (resp. $p_{\sigma_{X}D}^{\ast}=\left(1-\eta\right)\sigma_{X}$). At the corners lie the strategies ALLD (always defect), ALLC (always cooperate), TFT (tit-for-tat), and STFT (suspicious tit-for-tat).\label{fig:LTFT}}
\end{figure}

\subsection{Linear reactive learning strategies}
A pertinent question at this point is whether there is a ``natural" map from $\textbf{Mem}_{X}^{1}$ to $\textbf{RL}_{X}$. Let $\left(p_{0},\mathbf{p}_{\bullet\bullet}\right) =\left(p_{0},p_{CC},p_{CD},p_{DC},p_{DD}\right)$ be a memory-one strategy. If $\left(p_{0}',p^{\ast}\right)$ is the corresponding reactive learning strategy, then the first requirement we impose is $p_{0}'=p_{0}$. If $\sigma_{X}=1$, then $X$ plays $C$ with probability one. It is therefore reasonable to insist that $p_{1y}^{\ast}=p_{Cy}$. Similarly, $X$ plays $D$ with probability one when $\sigma_{X}=0$, and we insist that $p_{0y}^{\ast}=p_{Dy}$. Suppose now that $\sigma_{X}$ and $\sigma_{X}'$ are two mixed actions for $X$. If $Y$ plays $y\in\left\{C,D\right\}$, then the responses for $X$ corresponding to $\sigma_{X}$ and $\sigma_{X}'$ are $p_{\sigma_{X}y}^{\ast}$ and $p_{\sigma_{X}'y}^{\ast}$, respectively. If $X$ plays $\sigma_{X}$ with probability $w\in\left[0,1\right]$ and $\sigma_{X}'$ with probability $1-w$, then it is also natural to insist that the response is $p_{\sigma_{X}y}^{\ast}$ with probability $w$ and $p_{\sigma_{X}'y}^{\ast}$ with probability $1-w$. Thus, for any $\sigma_{X}\in\left[0,1\right]$ and $y\in\left\{C,D\right\}$, with these requirements $p^{\ast}$ can be written uniquely in terms of $\mathbf{p}_{\bullet\bullet}$ as
\begin{align}
p_{\sigma_{X}y}^{\ast} = \sigma_{X} p_{1y}^{\ast} + \left(1-\sigma_{X}\right) p_{0y}^{\ast} = \sigma_{X} p_{Cy} + \left(1-\sigma_{X}\right) p_{Dy} .
\end{align}
Using this map, one can naturally identify $\textbf{Mem}_{X}^{1}$ with the set of \textit{linear} reactive learning strategies, $\textbf{LRL}_{X}\subseteq\textbf{RL}_{X}$, consisting of those functions $p^{\ast}:\left[0,1\right]\times\left\{C,D\right\}\rightarrow\left[0,1\right]$ for which there exist $a,b,c,d\in\mathbb{R}$ with
\begin{subequations}
\begin{align}
p_{\sigma_{X}C}^{\ast} &= \sigma_{X}a + \left(1-\sigma_{X}\right) c ; \\
p_{\sigma_{X}D}^{\ast} &= \sigma_{X}b + \left(1-\sigma_{X}\right) d .
\end{align}
\end{subequations}
Clearly, any such $a,b,c,d$ must lie in $\left[0,1\right]$ since $p_{\sigma_{X}y}^{\ast}\in\left[0,1\right]$ for every $\sigma_{X}\in\left[0,1\right]$ and $y\in\left\{C,D\right\}$.

Under this correspondence, the strategy of Example~\ref{ex:eta} has parameters $\left(1/2,1,1-\eta ,\eta, 0\right)$. But note that this map, $\textbf{Mem}_{X}^{1}\rightarrow\textbf{RL}_{X}$, is not surjective due to the fact that not every reactive learning strategy is linear. For example, if $\left(a,b,c,d\right)\in\left[0,1\right]^{4}$ and $p^{\ast}\in\textbf{RL}_{X}$ is the quadratic response function defined by
\begin{subequations}
\begin{align}
p_{\sigma_{X}C}^{\ast}\coloneqq\left(\sigma_{X}\right)^{2}a+\left(1-\left(\sigma_{X}\right)^{2}\right) c ; \\
p_{\sigma_{X}D}^{\ast}\coloneqq\left(\sigma_{X}\right)^{2}b+\left(1-\left(\sigma_{X}\right)^{2}\right) d ,
\end{align}
\end{subequations}
then there exists no $\left(p_{CC},p_{CD},p_{DC},p_{DD}\right)\in\left[0,1\right]^{4}$ mapping to $p^{\ast}$ provided $a\neq c$ or $b\neq d$.

\subsection{Stationary distributions}
Suppose that $\left(p_{0},p^{\ast}\right)$ and $\left(q_{0},q^{\ast}\right)$ are reactive learning strategies for $X$ and $Y$, respectively. These strategies generate a Markov chain on the (infinite) space $\left\{C,D\right\}^{2}\times\left[0,1\right]^{2}$ with transition probabilities between $\Big(\left(x,y\right) ,\left(\sigma_{X},\sigma_{Y}\right)\Big) ,\Big(\left(x',y'\right) ,\left(p_{\sigma_{X}y}^{\ast},q_{\sigma_{Y}x}^{\ast}\right)\Big)\in\left\{C,D\right\}^{2}\times\left[0,1\right]^{2}$ given by
\begin{align}
P_{\Big(\left(x,y\right) ,\left(\sigma_{X},\sigma_{Y}\right)\Big)\rightarrow\Big(\left(x',y'\right) ,\left(p_{\sigma_{X}y}^{\ast},q_{\sigma_{Y}x}^{\ast}\right)\Big)} &\coloneqq 
\begin{cases}
p_{\sigma_{X}y}^{\ast} q_{\sigma_{Y}x}^{\ast} & x'=C,\ y'=C , \\
p_{\sigma_{X}y}^{\ast} \left(1-q_{\sigma_{Y}x}^{\ast}\right) & x'=C,\ y'=D , \\
\left(1-p_{\sigma_{X}y}^{\ast}\right) q_{\sigma_{Y}x}^{\ast} & x'=D,\ y'=C , \\
\left(1-p_{\sigma_{X}y}^{\ast}\right) \left(1-q_{\sigma_{Y}x}^{\ast}\right) & x'=D,\ y'=D .
\end{cases}
\end{align}
To simplify notation, we can also denote the right-hand side of this equation by $p_{\sigma_{X}y}^{\ast}\left(x'\right) q_{\sigma_{Y}x}^{\ast}\left(y'\right)$.

If $\nu$ is a stationary distribution of this chain, then, for any $\Big(\left(x,y\right) ,\left(\sigma_{X},\sigma_{Y}\right)\Big)\in\left\{C,D\right\}^{2}\times\left[0,1\right]^{2}$,
\begin{align}
\nu\Big(\left(x,y\right) ,\left(\sigma_{X},\sigma_{Y}\right)\Big) &= \int\limits_{\substack{\Big(\left(x',y'\right) ,\left(\sigma_{X}',\sigma_{Y}'\right)\Big) \\ \left(p_{\sigma_{X}'y'}^{\ast},q_{\sigma_{Y}'x'}^{\ast}\right) =\left(\sigma_{X},\sigma_{Y}\right)}} P_{\Big(\left(x',y'\right) ,\left(\sigma_{X}',\sigma_{Y}'\right)\Big)\rightarrow\Big(\left(x,y\right) ,\left(\sigma_{X},\sigma_{Y}\right)\Big)} \, d\nu\Big(\left(x',y'\right) ,\left(\sigma_{X}',\sigma_{Y}'\right)\Big) \nonumber \\
&= \int\limits_{\substack{\Big(\left(x',y'\right) ,\left(\sigma_{X}',\sigma_{Y}'\right)\Big) \\ \left(p_{\sigma_{X}'y'}^{\ast},q_{\sigma_{Y}'x'}^{\ast}\right) =\left(\sigma_{X},\sigma_{Y}\right)}} \sigma_{X}\left(x\right) \sigma_{Y}\left(y\right) \, d\nu\Big(\left(x',y'\right) ,\left(\sigma_{X}',\sigma_{Y}'\right)\Big) . \label{eq:recurrence}
\end{align}
In general, $\nu$ is difficult to give explicitly. However, it is possible to understand the marginal distributions on $\sigma_{X}$ and $\sigma_{Y}$ in more detail (see \ref{sec:convergenceOfMixedActions}). In any case, having an explicit formula for $\nu$ is not necessary for obtaining our main result on feasible payoff regions, which we turn to in the next section.

\subsection{Feasible payoff regions}\label{sec:feasibleRL}
By looking at the feasible region of a strategy, we uncover a nice relationship between a memory-one strategy, $\mathbf{p}$, and its corresponding (linear) reactive learning strategy, $\mathbf{p}^{\ast}$. Namely, for every $\mathbf{p}\in\textbf{Mem}_{X}^{1}$, we have $\mathcal{C}\left(\mathbf{p}^{\ast}\right)\subseteq\mathcal{C}\left(\mathbf{p}\right)$. In this section, we give a proof of this fact and illustrate some of its consequences.

For $t\geqslant 1$, let $\mathcal{H}_{t}=\left(\left\{C,D\right\}^{2}\right)^{t}$ be the history of play from time $0$ through time $t-1$ \citep{fudenberg:MIT:1991}. When $t=0$, $\mathcal{H}_{0}=\left\{\varnothing\right\}$, where $\varnothing$ denotes the ``empty" history, indicating that no play came before the present encounter. A behavioral strategy for a player specifies, for every possible history of play, a probability of using $C$ in the next encounter. That is, if $\mathcal{H}\coloneqq\sqcup_{t\geqslant 0}\mathcal{H}_{t}$, then a behavioral strategy is a map $\mathcal{H}\rightarrow\left[0,1\right]$. The following lemma shows that when considering the feasible region of a memory-one or reactive learning strategy, one can assume without a loss of generality that the opponent is playing a Markov strategy:

\begin{lemma}\label{lem:markov}
Let $\mathcal{M}\subseteq\mathcal{B}$ be the set of all Markov strategies, i.e.
\begin{align}
\mathcal{M} &\coloneqq \Big\{ \mathbf{y} : \left\{1,2,\dots\right\}\times\left\{C,D\right\}^{2}\rightarrow\left[0,1\right] \Big\} . \label{eq:markovStr}
\end{align}
For any $\mathbf{x}\in\textbf{Mem}_{X}^{1}\cup\textbf{RL}_{X}$, we have $\mathcal{C}\left(\mathbf{x}\right) =\left\{\left(\pi_{Y}\left(\mathbf{x},\mathbf{y}\right) ,\pi_{X}\left(\mathbf{x},\mathbf{y}\right)\right)\right\}_{\mathbf{y}\in\mathcal{M}}$.
\end{lemma}
\begin{proof}
When $\mathbf{p}\in\textbf{Mem}_{X}^{1}$, the lemma follows from \citep[][Appendix A]{press:PNAS:2012}. Specifically, when $X$ plays $\mathbf{p}\in\textbf{Mem}_{X}^{1}$ against $\mathbf{y}\in\mathcal{B}$, consider the time-$t$ distributions $\mu_{t}$ on $\left\{C,D\right\}^{2}$ and $\overline{\mu}_{t}$ on $\mathcal{H}_{t}$. For $\left(x_{t+1},y_{t+1}\right)\in\left\{C,D\right\}^{2}$,
\begin{align}
\mu_{t+1}\left(x_{t+1},y_{t+1}\right) &= \sum_{h_{t+1}\in\mathcal{H}_{t+1}} p_{x_{t}y_{t}}\left(x_{t+1}\right) \mathbf{y}_{h_{t+1}}\left(y_{t+1}\right) \overline{\mu}_{t+1}\left(h_{t+1}\right) \nonumber \\
&= \sum_{h_{t+1}\in\mathcal{H}_{t+1}} p_{x_{t}y_{t}}\left(x_{t+1}\right) \mathbf{y}_{\left(h_{t},\left(x_{t},y_{t}\right)\right)}\left(y_{t+1}\right) \overline{\mu}_{t+1}\left(h_{t+1}\right) \nonumber \\
&= \sum_{\left(x_{t},y_{t}\right)\in\left\{C,D\right\}^{2}} p_{x_{t}y_{t}}\left(x_{t+1}\right) \sum_{h_{t}\in\mathcal{H}_{t}} \mathbf{y}_{\left(h_{t},\left(x_{t},y_{t}\right)\right)}\left(y_{t+1}\right)\mu_{t}\left(x_{t},y_{t}\mid h_{t}\right)\overline{\mu}_{t}\left(h_{t}\right) .
\end{align}
Therefore, the same sequence of distributions $\left\{\mu_{t}\right\}_{t\geqslant 0}$ arises when $Y$ uses the Markov strategy defined by
\begin{align}
q_{x_{t}y_{t}}^{t+1}\left(y_{t+1}\right) &\coloneqq \frac{\sum_{h_{t}\in\mathcal{H}_{t}} \mathbf{y}_{\left(h_{t},\left(x_{t},y_{t}\right)\right)}\left(y_{t+1}\right)\mu_{t}\left(x_{t},y_{t}\mid h_{t}\right)\overline{\mu}_{t}\left(h_{t}\right)}{\sum_{h_{t}\in\mathcal{H}_{t}} \mu_{t}\left(x_{t},y_{t}\mid h_{t}\right)\overline{\mu}_{t}\left(h_{t}\right)} .
\end{align}

If $p^{\ast}:\left[0,1\right]\times\left\{C,D\right\}\rightarrow\left[0,1\right]$ is a reactive learning strategy that $X$ uses against $\mathbf{y}\in\mathcal{B}$, then for every $t\geqslant 0$ there are distributions $\nu_{t}$ on $\left\{C,D\right\}^{2}$, $\chi_{t}$ on $\left[0,1\right]$, and $\overline{\nu}_{t}$ on $\mathcal{H}_{t}\times\left[0,1\right]$. For $\left(x_{t+1},y_{t+1}\right)\in\left\{C,D\right\}^{2}$,
\begin{align}
&\scalebox{0.9}{$\nu_{t+1}\left(x_{t+1},y_{t+1}\right)$} \nonumber \\
&= \scalebox{0.9}{$\int\limits_{\left(h_{t+1},\sigma_{X}^{t}\right)\in\mathcal{H}_{t+1}\times\left[0,1\right]} p_{\sigma_{X}^{t}y_{t}}^{\ast}\left(x_{t+1}\right) \mathbf{y}_{h_{t+1}}\left(y_{t+1}\right) \, d\overline{\nu}_{t+1}\left(h_{t+1},\sigma_{X}^{t}\right)$} \nonumber \\
&= \scalebox{0.9}{$\sum\limits_{\left(x_{t},y_{t}\right)\in\left\{C,D\right\}^{2}}\int\limits_{\sigma_{X}^{t}\in\left[0,1\right]} p_{\sigma_{X}^{t}y_{t}}^{\ast}\left(x_{t+1}\right) \int\limits_{\left(h_{t},\sigma_{X}^{t-1}\right)\in\mathcal{H}_{t}\times\left[0,1\right]} \mathbf{y}_{\left(h_{t},\left(x_{t},y_{t}\right)\right)}\left(y_{t+1}\right) \, d\chi_{t}\left(\sigma_{X}^{t}\mid \left(h_{t},\left(x_{t},y_{t}\right)\right),\sigma_{X}^{t-1}\right)\, d\overline{\nu}_{t}\left(h_{t},\sigma_{X}^{t-1}\right)$} .
\end{align}
Consider the Markov strategy for $Y$ with $q_{0}\coloneqq y_{\varnothing}$ and $q_{x_{0}y_{0}}^{1}\left(y_{1}\right)\coloneqq\mathbf{y}_{\left(x_{0},y_{0}\right)}\left(y_{1}\right)$. For $t\geqslant 1$, let
\begin{align}
\scalebox{1.0}{$q_{x_{t}y_{t}}^{t+1}\left(y_{t+1}\right)$} &\coloneqq \scalebox{1.0}{$\frac{\int\limits_{\sigma_{X}^{t}\in\left[0,1\right]} p_{\sigma_{X}^{t}y_{t}}^{\ast}\left(x_{t+1}\right) \int\limits_{\left(h_{t},\sigma_{X}^{t-1}\right)\in\mathcal{H}_{t}\times\left[0,1\right]} \mathbf{y}_{\left(h_{t},\left(x_{t},y_{t}\right)\right)}\left(y_{t+1}\right) \, d\chi_{t}\left(\sigma_{X}^{t}\mid \left(h_{t},\left(x_{t},y_{t}\right)\right),\sigma_{X}^{t-1}\right)\, d\overline{\nu}_{t}\left(h_{t},\sigma_{X}^{t-1}\right)}{\int\limits_{\sigma_{X}^{t}\in\left[0,1\right]} p_{\sigma_{X}^{t}y_{t}}^{\ast}\left(x_{t+1}\right) \, d\chi_{t}\left(\sigma_{X}^{t}\mid x_{t},y_{t}\right) \nu_{t}\left(x_{t},y_{t}\right)}$} .
\end{align}
If $\nu_{t}'$ and $\chi_{t}'$ are the analogues of $\nu_{t}$ and $\chi_{t}$ for $p^{\ast}$ against $\left\{\mathbf{q}^{t}\right\}_{t\geqslant 1}$, then clearly $\nu_{t}=\nu_{t}'$ and $\chi_{t}=\chi_{t}'$ for $t=0,1$. Suppose that for some $t\geqslant 1$, we have $\nu_{t}=\nu_{t}'$ and $\chi_{t}=\chi_{t}'$. It follows, then, that at time $t+1$,
\begin{align}
\nu_{t+1}'\left(x_{t+1},y_{t+1}\right) &= \sum_{\left(x_{t},y_{t}\right)\in\left\{C,D\right\}^{2}} q_{x_{t}y_{t}}^{t+1}\left(y_{t+1}\right) \int\limits_{\sigma_{X}^{t}\in\left[0,1\right]} p_{\sigma_{X}^{t}y_{t}}^{\ast}\left(x_{t+1}\right) \, d\chi_{t}'\left(\sigma_{X}^{t}\mid x_{t},y_{t}\right) \, \nu_{t}'\left(x_{t},y_{t}\right) \nonumber \\
&= \sum_{\left(x_{t},y_{t}\right)\in\left\{C,D\right\}^{2}} q_{x_{t}y_{t}}^{t+1}\left(y_{t+1}\right) \int\limits_{\sigma_{X}^{t}\in\left[0,1\right]} p_{\sigma_{X}^{t}y_{t}}^{\ast}\left(x_{t+1}\right) \, d\chi_{t}\left(\sigma_{X}^{t}\mid x_{t},y_{t}\right) \, \nu_{t}\left(x_{t},y_{t}\right) \nonumber \\
&= \nu_{t+1}\left(x_{t+1},y_{t+1}\right) ,
\end{align}
which gives the desired result for $\mathbf{x}\in\textbf{RL}_{X}$.
\end{proof}

This lemma leads to a straightforward proof of our main result:
\begin{theorem}\label{thm:strategyGeometry}
$\mathcal{C}\left(\mathbf{p}^{\ast}\right)\subseteq\mathcal{C}\left(\mathbf{p}\right)$ for every $\mathbf{p}\in\textbf{Mem}_{X}^{1}$.
\end{theorem}
\begin{proof}
By Lemma~\ref{lem:markov}, for $\mathbf{x}\in\textbf{RL}_{X}$, we may assume the opponent's strategy is Markovian, meaning that it has a memory of one round into the past but can depend on the current round, $t$. This dependence on $t$ distinguishes a Markov strategy from a memory-one strategy, the latter of which also has memory of one round into the past but is independent of $t$. We denote by $\mathcal{M}$ the set of all Markov strategies (\eq{markovStr}).

Let $\mathbf{p}^{\ast}$ be a linear reactive learning strategy for $X$ and suppose that $\mathbf{y}\in\mathcal{M}$. For every $t\geqslant 0$, these strategies generate a distribution $\nu_{t}^{\ast}$ over $\left\{C,D\right\}^{2}\times\left[0,1\right]$. For any strategy $\mathbf{q}$ against $\mathbf{p}$, there is a sequence of distributions $\mu_{t}$ on $\left\{C,D\right\}^{2}$ generated by these two strategies. We prove the proposition by finding $\left\{\mathbf{q}^{t}\right\}_{t\geqslant 1}\in\mathcal{M}$ such that $\mu_{t}\left(x_{t},y_{t}\right) =\nu_{t}^{\ast}\left(\left\{\left(x_{t},y_{t}\right)\right\}\times\left[0,1\right]\right)$ for every $\left(x_{t},y_{t}\right)\in\left\{C,D\right\}^{2}$ and $t\geqslant 0$.

Let $\chi_{t}$ be the (marginal) distribution on $\sigma_{X}^{t}\in\left[0,1\right]$ at time $t$. For $y_{t}\in\left\{C,D\right\}$, denote by $\chi_{t}\left(\cdot\mid y_{t}\right)$ this distribution conditioned on $Y$ using action $y_{t}$ at time $t$. For $t\geqslant 0$, consider the strategy with $q_{0}\coloneqq y_{\varnothing}$ and
\begin{subequations}
\begin{align}
q_{Cy_{t}}^{t+1}\left(y_{t+1}\right) &\coloneqq \frac{\displaystyle\int\limits_{\sigma_{X}^{t}\in\left[0,1\right]} \sigma_{X}^{t} \left(\sigma_{X}^{t}y_{Cy_{t}}^{t+1}\left(y_{t+1}\right) +\left(1-\sigma_{X}^{t}\right) y_{Dy_{t}}^{t+1}\left(y_{t+1}\right)\right) \, d\chi_{t}\left(\sigma_{X}^{t}\mid y_{t}\right)}{\displaystyle\int\limits_{\sigma_{X}^{t}\in\left[0,1\right]}\sigma_{X}^{t}\,d\chi_{t}\left(\sigma_{X}^{t}\mid y_{t}\right)} ; \\
q_{Dy_{t}}^{t+1}\left(y_{t+1}\right) &\coloneqq \frac{\displaystyle\int\limits_{\sigma_{X}^{t}\in\left[0,1\right]} \left(1-\sigma_{X}^{t}\right) \left(\sigma_{X}^{t}y_{Cy_{t}}^{t+1}\left(y_{t+1}\right) +\left(1-\sigma_{X}^{t}\right) y_{Dy_{t}}^{t+1}\left(y_{t+1}\right)\right) \, d\chi_{t}\left(\sigma_{X}^{t}\mid y_{t}\right)}{\displaystyle\int\limits_{\sigma_{X}^{t}\in\left[0,1\right]}\left(1-\sigma_{X}^{t}\right)\,d\chi_{t}\left(\sigma_{X}^{t}\mid y_{t}\right)} .
\end{align}
\end{subequations}
Clearly, $\mu_{0}\left(x_{0},y_{0}\right) =\nu_{0}^{\ast}\left(\left\{\left(x_{0},y_{0}\right)\right\}\times\left[0,1\right]\right)$ for every $\left(x_{0},y_{0}\right)\in\left\{C,D\right\}^{2}$. Suppose, for some $t\geqslant$, that $\mu_{t}\left(x_{t},y_{t}\right) =\nu_{t}^{\ast}\left(\left\{\left(x_{t},y_{t}\right)\right\}\times\left[0,1\right]\right)$ for every $\left(x_{t},y_{t}\right)\in\left\{C,D\right\}^{2}$. For $\left(x_{t+1},y_{t+1}\right)\in\left\{C,D\right\}^{2}$, we then have
\begin{align}
\mu_{t+1}\left(x_{t+1},y_{t+1}\right) &= \sum_{\left(x_{t},y_{t}\right)\in\left\{C,D\right\}^{2}} p_{x_{t}y_{t}}\left(x_{t+1}\right) q_{x_{t}y_{t}}^{t+1}\left(y_{t+1}\right) \, \mu_{t}\left(x_{t},y_{t}\right) \nonumber \\
&= \sum_{y_{t}\in\left\{C,D\right\}} \left( p_{Cy_{t}}\left(x_{t+1}\right) q_{Cy_{t}}^{t+1}\left(y_{t+1}\right) \, \mu_{t}\left(C,y_{t}\right) + p_{Dy_{t}}\left(x_{t+1}\right) q_{Dy_{t}}^{t+1}\left(y_{t+1}\right) \, \mu_{t}\left(D,y_{t}\right) \right) \nonumber \\
&= \sum_{y_{t}\in\left\{C,D\right\}} p_{Cy_{t}}\left(x_{t+1}\right)\int\limits_{\sigma_{X}^{t}\in\left[0,1\right]} \sigma_{X}^{t} \left(\sigma_{X}^{t}y_{Cy_{t}}^{t+1}\left(y_{t+1}\right) +\left(1-\sigma_{X}^{t}\right) y_{Dy_{t}}^{t+1}\left(y_{t+1}\right)\right) \, d\chi_{t}\left(\sigma_{X}^{t}\mid y_{t}\right) \nonumber \\
&\quad + \sum_{y_{t}\in\left\{C,D\right\}} p_{Dy_{t}}\left(x_{t+1}\right)\int\limits_{\sigma_{X}^{t}\in\left[0,1\right]} \left(1-\sigma_{X}^{t}\right) \left(\sigma_{X}^{t}y_{Cy_{t}}^{t+1}\left(y_{t+1}\right) +\left(1-\sigma_{X}^{t}\right) y_{Dy_{t}}^{t+1}\left(y_{t+1}\right)\right) \, d\chi_{t}\left(\sigma_{X}^{t}\mid y_{t}\right) \nonumber \\
&= \sum_{\left(x_{t},y_{t}\right)\in\left\{C,D\right\}^{2}} y_{x_{t}y_{t}}^{t+1}\left(y_{t+1}\right) \int\limits_{\sigma_{X}^{t}\in\left[0,1\right]} \left(\sigma_{X}^{t}p_{Cy_{t}}+\left(1-\sigma_{X}^{t}\right) p_{Dy_{t}}\right) \, d\nu_{t}^{\ast}\left(\left\{\left(x_{t},y_{t}\right)\right\}\times\left\{\sigma_{X}^{t}\right\}\right) \nonumber \\
&= \sum_{\left(x_{t},y_{t}\right)\in\left\{C,D\right\}^{2}} \int\limits_{\sigma_{X}^{t}\in\left[0,1\right]} p_{\sigma_{X}^{t}y_{t}}^{\ast}\left(x_{t+1}\right) y_{x_{t}y_{t}}^{t+1}\left(y_{t+1}\right) \, d\nu_{t}^{\ast}\left(\left\{\left(x_{t},y_{t}\right)\right\}\times\left\{\sigma_{X}^{t}\right\}\right) \nonumber \\
&= \nu_{t+1}^{\ast}\left(\left\{\left(x_{t+1},y_{t+1}\right)\right\}\times\left[0,1\right]\right) .
\end{align}
Therefore, by induction and the definition of expected payoff in an iterated game, $\mathcal{C}\left(\mathbf{p}^{\ast}\right)\subseteq\mathcal{C}\left(\mathbf{p}\right)$.
\end{proof}

As a consequence of Theorem~\ref{thm:strategyGeometry}, we see that $\mathbf{p}^{\ast}$ is a enforces a linear payoff relationship \citep{press:PNAS:2012} whenever $\mathbf{p}$ does. However, the converse need not hold; \fig{regionSubset}\textit{(b)} gives an example in which $X$'s payoff is a function of $Y$'s when $X$ uses $\mathbf{p}^{\ast}$ but not when $X$ uses $\mathbf{p}$. Although this example illustrates an extreme case of when the payoff region collapses, perhaps the most interesting behavior is illustrated by \fig{regionSubset}\textit{(a),(c),(d)}. In these examples, we focus on the payoff regions that can be obtained against memory-one opponents. Using $\mathbf{p}^{\ast}$ instead of $\mathbf{p}$ can both bias payoffs in favor of $X$ and limit potential losses against a spiteful opponent.

\begin{figure}
\centering
\includegraphics[width=0.8\textwidth]{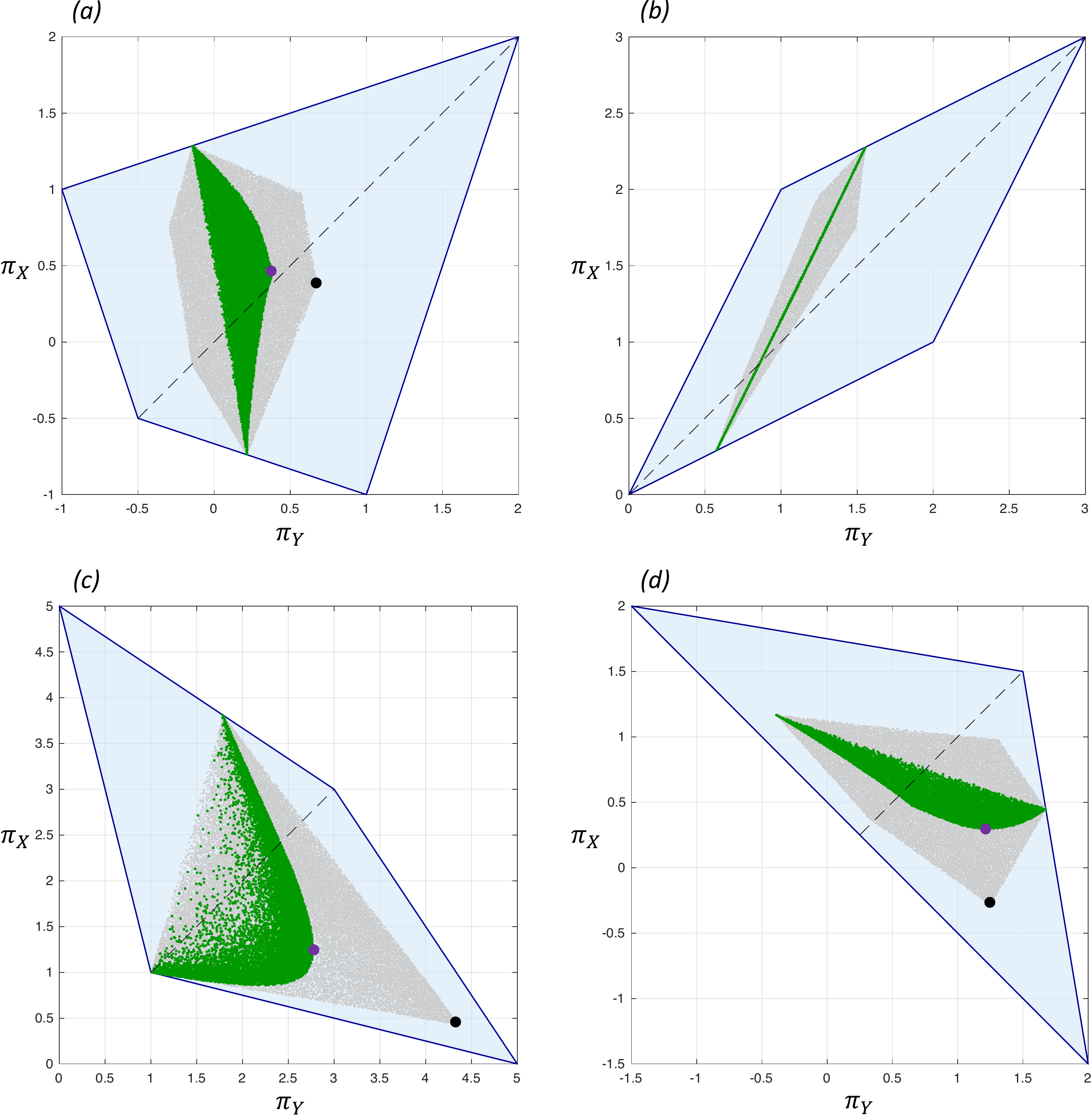}
\caption{\small Simulated payoffs against a fixed memory-one strategy, $\mathbf{p}$ (grey), and its corresponding reactive learning strategy, $\mathbf{p}^{\ast}$ (green), as the opponent plays $10^{5}$ randomly-chosen strategies $\mathbf{q}\in\textbf{Mem}_{X}^{1}$. \textit{(a)} If the opponent is greedy and wishes to optimize his or her own payoff only, then upon exploring the space $\textbf{Mem}_{X}^{1}$ for sufficiently long, the payoffs will end up at the black point when $X$ uses $\mathbf{p}$ and at the magenta point when $X$ uses $\mathbf{p}^{\ast}$. In this scenario, $\mathbf{p}$ favors $Y$ having a higher payoff than $X$, while $\mathbf{p}^{\ast}$ favors $X$ having a higher payoff than $Y$. Thus, $\mathbf{p}^{\ast}$ extorts a payoff-maximizing opponent while $\mathbf{p}$ is more generous. \textit{(b)} The payoffs against $\mathbf{p}^{\ast}$ (green) can fall along a line even when those against $\mathbf{p}$ (grey) form a two-dimensional region. In \textit{(c)}, by using $\mathbf{p}^{\ast}$ instead of $\mathbf{p}$, $X$ can limit the payoff the opponent receives from the black point to the magenta point. Similarly, in \textit{(d)}, $X$ can limit the potential ``punishment" incurred from $Y$. When $X$ uses $\mathbf{p}$, the opponent can choose a strategy that gives $X$ a negative payoff (black point). When $X$ uses $\mathbf{p}^{\ast}$, no such strategy of the opponent exists, and the worst payoff $X$ can possibly receive is positive (magenta point). The parameters used are \textit{(a)} $\mathbf{p}=\left(0.90,0.50,0.01,0.20,0.90\right)$ and $R=2$, $S=-1$, $T=1$, and $P=1/2$; \textit{(b)} $\mathbf{p}=\left(1.0000,0.6946,0.0354,0.1168,0.3889\right)$ and $R=3$, $S=1$, $T=2$, and $P=0$; \textit{(c)} $\mathbf{p}=\left(0.8623,0.6182,0.9528,0.5601,0.0001\right)$ and $R=3$, $S=0$, $T=5$, and $P=1$; and \textit{(d)} $\mathbf{p}=\left(0.5626,0.2381,0.7236,0.9537,0.1496\right)$ and $R=1/2$, $S=-3/2$, $T=2$, and $P=3/2$. Each coordinate of $\mathbf{q}$ is chosen independently from an arcsine (i.e. $\textrm{Beta}\left(1/2,1/2\right)$) distribution.\label{fig:regionSubset}}
\end{figure}

For a memory-one strategy $\mathbf{p}\in\textbf{Mem}_{X}^{1}$, we can ask how the region $\left\{\left(\pi_{Y}\left(\mathbf{p},\mathbf{q}\right) ,\pi_{X}\left(\mathbf{p},\mathbf{q}\right)\right)\right\}_{\mathbf{q}\in\textbf{Mem}_{X}^{1}}$ compares to $\left\{\left(\pi_{Y}\left(\mathbf{p}^{\ast},\mathbf{q}^{\ast}\right) ,\pi_{X}\left(\mathbf{p}^{\ast},\mathbf{q}^{\ast}\right)\right)\right\}_{\mathbf{q}\in\textbf{Mem}_{X}^{1}}$. In other words, does the map $\mathbf{p}\mapsto\mathbf{p}^{\ast}$ transform the feasible region of a strategy when the opponents are also subjected to this map? \fig{payoffDistortion} demonstrates that this map can significantly distort the distribution of payoffs within the feasible region.

\begin{figure}
\centering
\includegraphics[width=0.8\textwidth]{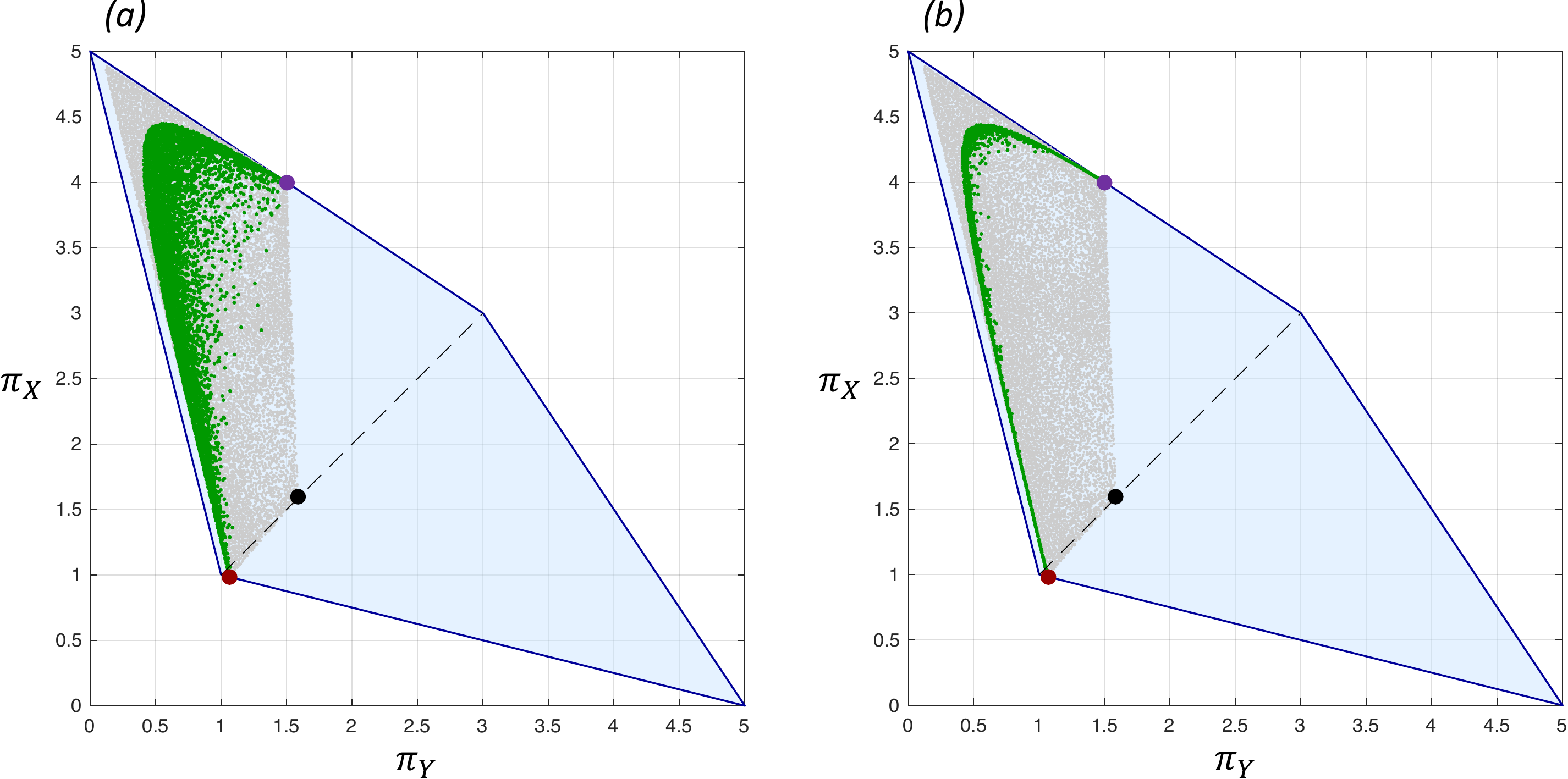}
\caption{\small Distortions in the distribution of payoffs against reactive learning strategies. In both panels, the grey region is formed by playing $10^{5}$ randomly-chosen strategies $\mathbf{q}\in\textbf{Mem}_{X}^{1}$ against a fixed strategy $\mathbf{p}\in\textbf{Mem}_{X}^{1}$. The green region in \textit{(a)} arises from simulating the payoffs of $\mathbf{p}^{\ast}$ against $10^{5}$ strategies $\mathbf{q}\in\textbf{Mem}_{X}^{1}$. In \textit{(b)}, this same reactive learning strategy, $\mathbf{p}^{\ast}$, is simulated against $10^{5}$ strategies $\mathbf{q}^{\ast}\in\textbf{RL}_{X}$ for $\mathbf{q}\in\textbf{Mem}_{X}^{1}$. In both panels, the optimal outcome for $Y$ is the black point when $X$ uses $\mathbf{p}$ and the magenta point when $X$ uses $\mathbf{p}^{\ast}$. The magenta point represents a much better outcome for $X$ and only a slightly worse outcome for $Y$ than the black point, indicating that $\mathbf{p}^{\ast}$ is highly extortionate relative to $\mathbf{p}$ when played against a payoff-maximizing opponent. In both panels, the parameters are $\mathbf{p}=\left(0.50,0.99,0.40,0.01,0.01\right)$ and $R=3$, $S=0$, $T=5$, and $P=1$. Each coordinate of $\mathbf{q}$ is chosen independently from an arcsine (i.e. $\textrm{Beta}\left(1/2,1/2\right)$) distribution.\label{fig:payoffDistortion}}
\end{figure}

\subsection{Optimization through mutation}
Suppose that $X$ uses a fixed reactive learning strategy, $\mathbf{p}^{\ast}$, for some $\mathbf{p}\in\textrm{Mem}_{X}^{1}$. Starting from some random memory-one strategy, $\mathbf{q}$, the opponent might seek to optimize his or her payoff through a series of mutations. In other words, $Y$ is subjected to the following process: First, sample a new strategy $\mathbf{q}'\in\textbf{Mem}_{X}^{1}$. If the payoff to $Y$ for $\mathbf{q}'$ against $\mathbf{p}^{\ast}$ exceeds that of $\mathbf{q}$ against $\mathbf{p}^{\ast}$, switch to $\mathbf{q}'$; otherwise, retain $\mathbf{q}$. This step then repeats until $Y$ has a sufficiently high payoff (or else has not changed strategies in some fixed number of steps). From \fig{payoffDistortion}, one expects this process to give different results from the same update scheme when $X$ plays the memory-one strategy $\mathbf{p}$ instead of $\mathbf{p}^{\ast}$.

As expected, \fig{optimization} shows that this optimization process behaves quite differently against $\mathbf{p}^{\ast}$ as it does against $\mathbf{p}$. Whereas using $\mathbf{p}$ in this example results in equitable outcomes, using $\mathbf{p}^{\ast}$ gives $X$ a much higher payoff than $Y$, indicating extortionate behavior. One can also imagine other optimization procedures (not covered here), such as when $\mathbf{q}'$ is always sufficiently close to $\mathbf{q}$ (i.e. local mutations). When $X$ uses $\mathbf{p}^{\ast}$, a path from the red point to the magenta point in \fig{payoffDistortion} through random local sampling of $\mathbf{q}$ typically requires $Y$ to initially accept lower payoffs. If $Y$ uses $\mathbf{q}^{\ast}$ instead of $\mathbf{q}$, as in \fig{payoffDistortion}\textit{(b)}, this effect is amplified.

\begin{figure}
\centering
\includegraphics[width=0.8\textwidth]{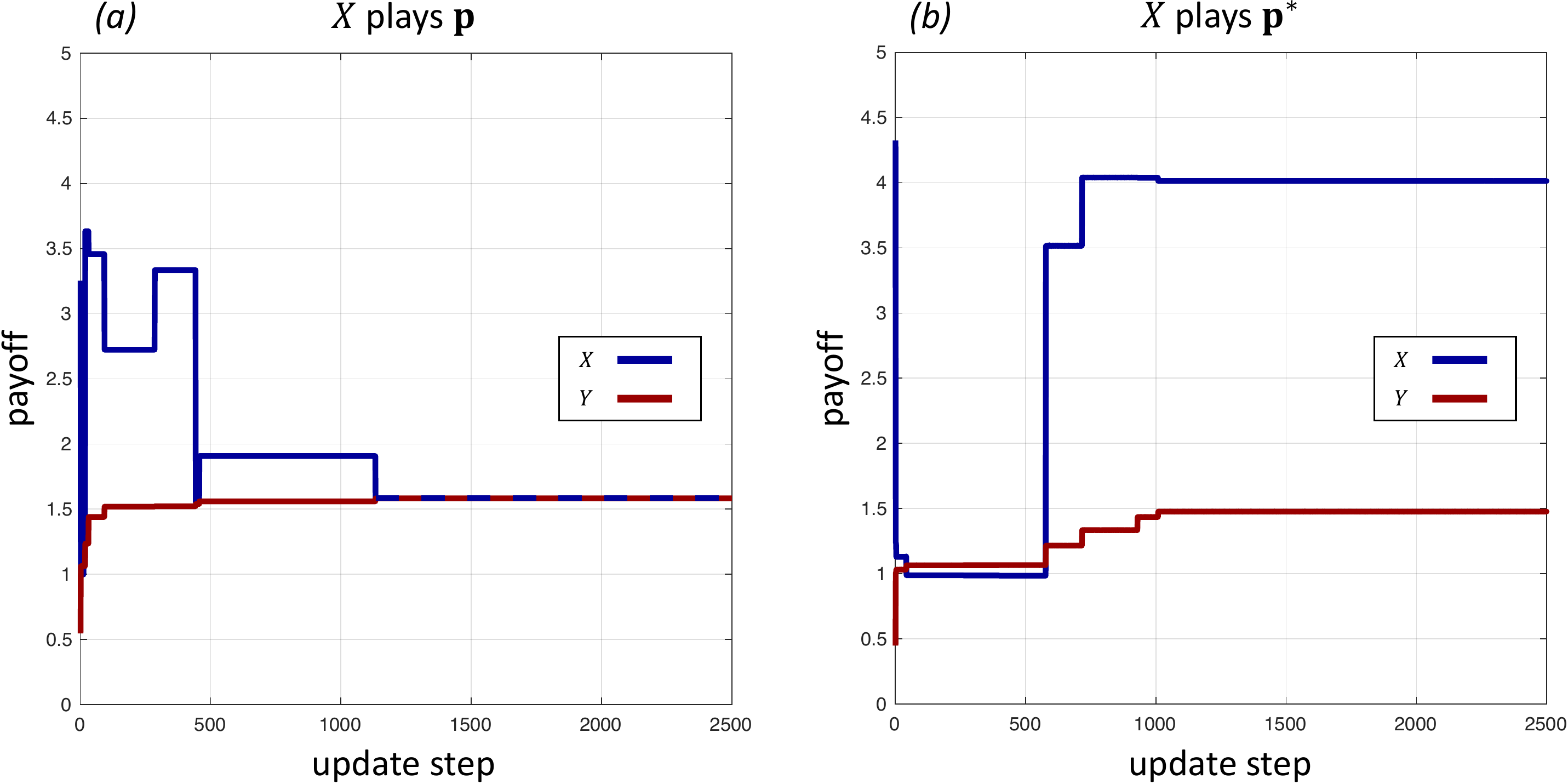}
\caption{\small Optimization against a memory-one strategy, \textit{(a)}, and the corresponding reactive learning strategy, \textit{(b)}. In each panel, $X$'s strategy is fixed with parameters $\mathbf{p}=\left(0.50,0.99,0.40,0.01,0.01\right)$. $Y$ chooses an initial memory-one strategy, $\mathbf{q}$, from an arcsine distribution. At each update step, $Y$ samples another strategy, $\mathbf{q}'$, from the same distribution. If $Y$'s payoff for playing $\mathbf{q}'$ against $X$ exceeds that of playing $\mathbf{q}$ against $X$, then $Y$ replaces his or her current strategy with $\mathbf{q}'$. Otherwise, $\mathbf{q}'$ is discarded and $Y$ retains $\mathbf{q}$. Over time, this process generates a sequence of payoff pairs for $X$ and $Y$, shown in \textit{(a)} and \textit{(b)}. Relative to $\mathbf{p}$, the reactive learning strategy $\mathbf{p}^{\ast}$ is highly extortionate.\label{fig:optimization}}
\end{figure}

\section{Discussion}
Our primary focus has been on the feasible region generated by a fixed strategy. This approach to studying $X$'s strategy is inspired by the ``zero-determinant" strategies of \citet{press:PNAS:2012}, which enforce linear subsets of the feasible region. This perspective has also been expanded to cover so-called ``partner" and ``rival" strategies \citep{akin:Games:2015,hilbe:GEB:2015,hilbe:NHB:2018}, which have proven extremely useful in understanding repeated games from an evolutionary perspective. The feasible region of a memory-one strategy, $\mathbf{p}$, is quite simple and can be characterized as the convex hull of at most $11$ points. Furthermore, these points are all straightforward to write down explicitly in terms of the payoff matrix and the entries of $\mathbf{p}$ (see \eq{elevenPoints}). The feasible region of a reactive learning strategy, in terms of its boundary and extreme points, is evidently more complicated in general.

Both memory-one and reactive learning strategies contain the set of all reactive strategies. For every memory-one strategy, $\mathbf{p}$, there exists a corresponding linear reactive learning strategy, $\mathbf{p}^{\ast}$, and this correspondence defines an injective map $\textbf{Mem}_{X}^{1}\rightarrow\textbf{RL}_{X}$. In general, however, $\mathbf{p}$ cannot be identified with its image, $\mathbf{p}^{\ast}$, unless $\mathbf{p}$ is reactive. We make this claim formally using the geometry of a strategy within the feasible region, $\mathcal{C}\left(\mathbf{p}\right)$, which captures all possible payoff pairs against an opponent. For any memory-one strategy, we have $\mathcal{C}\left(\mathbf{p}^{\ast}\right)\subseteq\mathcal{C}\left(\mathbf{p}\right)$. Therefore, reactive learning strategies generally allow a player to impose greater control over where payoffs fall within the feasible region than do traditional memory-one strategies. As illustrated in \fig{regionSubset}\textit{(a)}, this added control can prevent a greedy, self-payoff-maximizing opponent from obtaining more than $X$ when $X$ uses $\mathbf{p}^{\ast}$, even when such an opponent receives an unfair share of the payoffs when $X$ uses $\mathbf{p}$ instead. The proof of the containment $\mathcal{C}\left(\mathbf{p}^{\ast}\right)\subseteq\mathcal{C}\left(\mathbf{p}\right)$ also extends to discounted games, where each payoff unit received $t$ rounds into the future is valued at $\delta^{t}$ units at present for some ``discounting factor," $\delta\in\left[0,1\right]$.

Another property of the map $\textbf{Mem}_{X}^{1}\rightarrow\textbf{RL}_{X}$ sending $\mathbf{p}$ to $\mathbf{p}^{\ast}$ is that it distorts the distribution of payoffs within the feasible region. Since $\textbf{Mem}_{X}^{1}$ can be identified with the space of linear reactive learning strategies under this map, it is natural to compare the region of possible payoffs when $\mathbf{p}$ plays against memory-one strategies to the one obtained from when $\mathbf{p}^{\ast}$ plays against linear reactive learning strategies. These distortions, as illustrated in \fig{payoffDistortion}, are particularly relevant when $X$ plays against an opponent who is using a process such as simulated annealing to optimize payoff. One can see from this example that if $Y$ initially has a low payoff, then with localized strategy exploration they must be willing to accept lower payoffs before they find a strategy that improves their initial payoff. This concern is not relevant when $Y$ can simply compute the best response to $X$'s strategy, but it is highly pertinent to evolutionary settings in which the opponent's strategy is obtained through mutation and selection rather than ``computation."

Reactive learning strategies are also more intuitive than memory-one strategies in some ways. Rather than being a dictionary of mixed actions based on all possible observed outcomes, a reactive learning strategy is simply an algorithm for updating one's tendency to choose a certain action. It therefore allows a player to alter their behavior (mixed action) over time in response to various stimuli (actions of the opponent). This strategic approach to iterated games is reminiscent of both the Bush-Mosteller model \citep{bush:AMS:1953} and the weighted majority algorithm \citep{littlestone:ASFCS:1989}, although traditionally these models are not studied through the payoff regions they generate in iterated games. There are several interesting directions for future research in this area. For one, we have mainly considered the space of linear reactive learning strategies, but the space $\textbf{RL}_{X}$ is much larger and could potentially exhibit complicated evolutionary dynamics. Furthermore, one could relax the condition that these strategies be reactive and allow them to use $X$'s realized action in addition to $X$'s mixed action. But even without these complications, we have seen that linear reactive learning strategies have quite interesting relationships to traditional memory-one strategies.

\setcounter{section}{0}
\renewcommand{\thesection}{Appendix}
\renewcommand{\thesubsection}{\Alph{section}.\arabic{subsection}}
\renewcommand{\theequation}{\Alph{section}\arabic{equation}}
\renewcommand{\thefigure}{\Alph{section}\arabic{figure}}

\setcounter{equation}{0}
\setcounter{figure}{0}
\section{Convergence of mixed actions}\label{sec:convergenceOfMixedActions}
Suppose that $X$ and $Y$ use strategies $\left(p_{0},p^{\ast}\right)$ and $\left(q_{0},q^{\ast}\right)$, respectively. Let $\sigma_{X}^{0}=p_{0}$ and $\sigma_{Y}^{0}=q_{0}$ be the initial distributions on $\left\{C,D\right\}$ for $X$ and $Y$, respectively. If these distributions are known at time $t\geqslant 0$, then, \textit{on average}, the corresponding distributions at time $t+1$ are given by the system of equations,
\begin{subequations}\label{eq:starSystem}
\begin{align}
\sigma_{X}^{t+1} &\coloneqq \sigma_{Y}^{t}p_{\sigma_{X}^{t}C}^{\ast} + \left(1-\sigma_{Y}^{t}\right) p_{\sigma_{X}^{t}D}^{\ast} ; \\
\sigma_{Y}^{t+1} &\coloneqq \sigma_{X}^{t}q_{\sigma_{Y}^{t}C}^{\ast} + \left(1-\sigma_{X}^{t}\right) q_{\sigma_{Y}^{t}D}^{\ast} .
\end{align}
\end{subequations}
This system suggests a fixed-point analysis to determine whether the sequence $\left\{\left(\sigma_{X}^{t},\sigma_{Y}^{t}\right)\right\}_{t\geqslant 0}$ converges.

Suppose that $\left(\sigma_{X},\sigma_{Y}\right)\in\left[0,1\right]^{2}$ is a fixed point of this system, i.e.
\begin{subequations}
\begin{align}
\sigma_{X} &= \sigma_{Y}p_{\sigma_{X}C}^{\ast} + \left(1-\sigma_{Y}\right) p_{\sigma_{X}D}^{\ast} ; \\
\sigma_{Y} &= \sigma_{X}q_{\sigma_{Y}C}^{\ast} + \left(1-\sigma_{X}\right) q_{\sigma_{Y}D}^{\ast} .
\end{align}
\end{subequations}
We consider this system for two types of linear reactive learning strategies: those coming from reactive strategies and those coming from general memory-one strategies under the map $\textbf{Mem}_{X}^{1}\rightarrow\textbf{RL}_{X}$.

We first consider reactive strategies of the form $\left(p_{C},p_{D}\right)$, where $p_{C}$ (resp. $p_{D}$) is the probability a player uses $C$ after the opponent played $C$ (resp. $D$). Let $\left(p_{C},p_{D}\right)$ and $\left(q_{C},q_{D}\right)$ be fixed strategies for $X$ and $Y$. For these reactive strategies, the system \eq{starSystem} takes the form
\begin{subequations}\label{eq:reactiveSystem}
\begin{align}
\sigma_{X}^{t+1} &\coloneqq \sigma_{Y}^{t}p_{C} + \left(1-\sigma_{Y}^{t}\right) p_{D} ; \\
\sigma_{Y}^{t+1} &\coloneqq \sigma_{X}^{t}q_{C} + \left(1-\sigma_{X}^{t}\right) q_{D} .
\end{align}
\end{subequations}
One can easily check that this dynamical system has a unique fixed point, which \citet{hofbauer:CUP:1998} refer to as the ``asymptotic $C$-level" of $\left(p_{C},p_{D}\right)$ against $\left(q_{C},q_{D}\right)$, and which is given explicitly by
\begin{subequations}\label{eq:reactiveFPexplicit}
\begin{align}
\sigma_{X} &= \frac{p_{C}q_{D} + p_{D}\left(1-q_{D}\right)}{1-\left(p_{C}-p_{D}\right)\left(q_{C}-q_{D}\right)} ; \\
\sigma_{Y} &= \frac{p_{D}q_{C} + \left(1-p_{D}\right) q_{D}}{1-\left(p_{C}-p_{D}\right)\left(q_{C}-q_{D}\right)} .
\end{align}
\end{subequations}

Furthermore, we have the following, straightforward convergence result:
\begin{proposition}\label{prop:reactiveFP}
If $\left(p_{C},p_{D}\right) ,\left(q_{C},q_{D}\right)\in\left(0,1\right)^{2}$, and if $\left(\sigma_{X},\sigma_{Y}\right)\in\left(0,1\right)^{2}$ is given by \eq{reactiveFPexplicit}, then
\begin{align}
\lim_{t\rightarrow\infty}\left(\sigma_{X}^{t},\sigma_{Y}^{t}\right) =\left(\sigma_{X},\sigma_{Y}\right)
\end{align}
for any initial condition, $\left(p_{0},q_{0}\right)\in\left[0,1\right]^{2}$.
\end{proposition}
\begin{proof}
For $\left(p_{C},p_{D}\right) ,\left(q_{C},q_{D}\right)\in\left(0,1\right)^{2}$, consider the map
\begin{align}
f &: \left[ 0,1 \right]^{2} \longrightarrow \left[ 0,1 \right]^{2} \nonumber \\
&: \begin{pmatrix} x \\ y \end{pmatrix} \longmapsto \begin{pmatrix} yp_{C} + \left(1-y\right) p_{D} \\ xq_{C} + \left(1-x\right)q_{D} \end{pmatrix} .
\end{align}
For $\left(x,y\right) ,\left(x',y'\right)\in\left[0,1\right]^{2}$, we have
\begin{align}
f\left(x,y\right) -f\left(x',y'\right) &= \begin{pmatrix}\left(y-y'\right)\left(p_{C}-p_{D}\right) \\ \left(x-x'\right)\left(q_{C}-q_{D}\right)\end{pmatrix} .
\end{align}
It follows that $\left\| f\left(x,y\right) -f\left(x',y'\right)\right\|\leqslant\lambda \left\|\left(x,y\right) -\left(x',y'\right)\right\|$, where $\lambda\coloneqq\max\left\{\left| p_{C}-p_{D}\right| ,\left| q_{C}-q_{D}\right|\right\}<1$. By the contraction mapping theorem, there is then a unique fixed point $\left(\sigma_{X},\sigma_{Y}\right)\in\left[0,1\right]^{2}$ such that
\begin{align}
\lim_{t\rightarrow\infty} f^{t}\left(p_{0},q_{0}\right) &= \left(\sigma_{X},\sigma_{Y}\right)
\end{align}
for any $\left(p_{0},q_{0}\right)\in\left[0,1\right]^{2}$. It is straightforward to check that \eq{reactiveFPexplicit} is a fixed point of \eq{reactiveSystem}.
\end{proof}

In particular, if $\mu\coloneqq\left(\sigma_{X}\sigma_{Y},\sigma_{X}\left(1-\sigma_{Y}\right) ,\left(1-\sigma_{X}\right)\sigma_{Y},\left(1-\sigma_{X}\right)\left(1-\sigma_{Y}\right)\right)$, then a straightforward calculation shows that $\mu$ is the stationary distribution of $M\left(\left(p_{C},p_{D},p_{C},p_{D}\right) ,\left(q_{C},q_{D},q_{C},q_{D}\right)\right)$ (\eq{memOneMatrix}).

\begin{remark}
Proposition~\ref{prop:reactiveFP} need not hold if $p_{y}$ and $q_{x}$ are not strictly between $0$ and $1$. For example, when $X$ and $Y$ both play TFT, $f$ is a simple involution with $f\left(x,y\right) =\left(y,x\right)$, which preserves distance.
\end{remark}

Consider now the case of general memory-one strategies with $\mathbf{p}_{\bullet\bullet}\coloneqq\left(p_{CC},p_{CD},p_{DC},p_{DD}\right)$ for $X$ and $\mathbf{q}_{\bullet\bullet}\coloneqq\left(q_{CC},q_{CD},q_{DC},q_{DD}\right)$ for $Y$. For these strategies, the system defined by \eq{starSystem} has the form
\begin{subequations}\label{eq:memOneSystem}
\begin{align}
\sigma_{X}^{t+1} &\coloneqq \sigma_{Y}^{t}\left(\sigma_{X}^{t}p_{CC}+\left(1-\sigma_{X}^{t}\right) p_{DC}\right) + \left(1-\sigma_{Y}^{t}\right) \left(\sigma_{X}^{t}p_{CD}+\left(1-\sigma_{X}^{t}\right) p_{DD}\right) ; \\
\sigma_{Y}^{t+1} &\coloneqq \sigma_{X}^{t}\left(\sigma_{Y}^{t}q_{CC}+\left(1-\sigma_{Y}^{t}\right) q_{DC}\right) + \left(1-\sigma_{X}^{t}\right)\left(\sigma_{Y}^{t}q_{CD}+\left(1-\sigma_{Y}^{t}\right) q_{DD}\right) .
\end{align}
\end{subequations}
In the spirit of Proposition~\ref{prop:reactiveFP}, for fixed $\mathbf{p}_{\bullet\bullet},\mathbf{q}_{\bullet\bullet}\in\left(0,1\right)^{4}$, we could consider the map
\begin{align}
F &: \left[ 0,1 \right]^{2} \longrightarrow \left[ 0,1 \right]^{2} \nonumber \\
&: \begin{pmatrix} x \\ y \end{pmatrix} \longmapsto \begin{pmatrix} y\left(xp_{CC}+\left(1-x\right) p_{DC}\right) + \left(1-y\right) \left(xp_{CD}+\left(1-x\right) p_{DD}\right) \\ x\left(yq_{CC}+\left(1-y\right) q_{DC}\right) + \left(1-x\right)\left(yq_{CD}+\left(1-y\right) q_{DD}\right) \end{pmatrix}
\end{align}
and analyze its fixed points. At this point, however, a couple of remarks are in order:
\begin{enumerate}

\item[\textit{(i)}] $F$ need not be a contraction, even when $\mathbf{p}_{\bullet\bullet}$ and $\mathbf{q}_{\bullet\bullet}$ have entries strictly between $0$ and $1$. For example, with $\mathbf{p}_{\bullet\bullet}=\left(0.9566,0.2730,0.0056,0.0095\right)$ and $\mathbf{q}_{\bullet\bullet}=\left(0.9922,0.0918,0.3217,0.0054\right)$,
\begin{align}
0.0441 &= \left\| F\left(0.7404,0.6928\right) - F\left(0.8241,0.8280\right) \right\| \nonumber \\
&> \left\| \left(0.7404,0.6928\right) - \left(0.8241,0.8280\right) \right\| = 0.0253 .
\end{align}
We would conjecture that this map is an \textit{eventual} contraction, in which case the convergence result of Proposition~\ref{prop:reactiveFP} still holds (although the explicit formulas for $\sigma_{X}$ and $\sigma_{Y}$ differ from \eq{reactiveFPexplicit}).

\item[\textit{(ii)}] a fixed point of $F$, $\left(\sigma_{X},\sigma_{Y}\right)$, even when it exists and is unique, generally does not have the property that $\mu\left(\mathbf{p},\mathbf{q}\right) =\left(\sigma_{X}\sigma_{Y},\sigma_{X}\left(1-\sigma_{Y}\right) ,\left(1-\sigma_{X}\right)\sigma_{Y},\left(1-\sigma_{X}\right)\left(1-\sigma_{Y}\right)\right)$, where $\mu$ is the stationary distribution of \eq{memOneMatrix}. Furthermore, the long-run mean-frequency distribution on $\left\{C,D\right\}^{2}$ can be distinct from \textit{both} of these distributions, including when the opponent plays $\mathbf{q}$ against $\mathbf{p}^{\ast}$ and when they play $\mathbf{q}^{\ast}$ against $\mathbf{p}^{\ast}$. An example of when these four distributions are pairwise distinct is easy to write down, e.g. $\mathbf{p}=\left(0.01,0.01,0.01,0.99,0.01\right)$ and $\mathbf{q}=\left(0.99,0.99,0.01,0.99,0.99\right)$. All four distributions coincide when $\mathbf{p}$ and $\mathbf{q}$ are both reactive, but in general they can be distinct.

\end{enumerate}

\section*{Acknowledgments}
The authors are grateful to Krishnendu Chatterjee, Christian Hilbe, and Joshua Plotkin for many helpful conversations and for feedback on earlier versions of this work.

\section*{Funding statement}
The authors gratefully acknowledge support from the Lifelong Learning Machines program from DARPA/MTO. Research was sponsored by the Army Research Laboratory (ARL) and was accomplished under Cooperative Agreement Number W911NF-18-2-0265. The views and conclusions contained in this document are those of the authors and should not be interpreted as representing the official policies, either expressed or implied, of the Army Research Laboratory or the U.S. Government. The U.S. Government is authorized to reproduce and distribute reprints for Government purposes notwithstanding any copyright notation herein.

\end{document}